\documentclass[sigconf]{acmart}
\AtBeginDocument{%
	\providecommand\BibTeX{{%
			\normalfont B\kern-0.5em{\scshape i\kern-0.25em b}\kern-0.8em\TeX}}}

\newcommand{\stitle}[1]{\vspace{1ex} \noindent{\bf #1}}
\usepackage{subfigure}
\usepackage{multirow}
\usepackage{multicol}
\usepackage{booktabs} % for table \toprule、\midrule、\bottomrule

\usepackage{amssymb}
\usepackage{amsmath}

\newcommand{\concept}{\emph{Motif Resident}}

\newcommand{\smallconcept}{motif resident}

\newtheorem{theorem}{Theorem}
\newtheorem{definition}{Definition}

\newtheorem{lemma}{Lemma}

\newtheorem{fact}{Fact}

\usepackage{subfigure}
\usepackage[shortlabels]{enumitem} 
\usepackage{multirow}
\usepackage[noend]{algpseudocode}
\usepackage{algorithmicx,algorithm}

\begin{document}

\title{PSMC: Provable and Scalable Algorithms for Motif Conductance Based Graph Clustering}

\author{Longlong Lin}
\affiliation{%
  \institution{College of Computer and Information Science, Southwest University}
  \city{Chongqing}
  \country{China}
}
\email{longlonglin@swu.edu.cn}
\author{Tao Jia}
\affiliation{%
	\institution{College of Computer and Information Science, Southwest University}
  \city{Chongqing}
\country{China}
}
\email{tjia@swu.edu.cn}
\author{Zeli Wang}
\affiliation{%
	\institution{Chongqing University of Posts and
		Telecommunications}
  \city{Chongqing}
\country{China}
}
\email{zlwang@cqupt.edu.cn}
\author{Jin Zhao}
\affiliation{%
	\institution{School of Computer Science and Technology, HuaZhong University of Science and Technology}
  \city{Wuhan}
\country{China}
}
\email{zjin@hust.edu.cn}
\author{Rong-Hua Li}
\affiliation{%
	\institution{Shenzhen Institute of Technology}
  \city{Shenzhen}
\country{China}
}
\affiliation{%
	\institution{Beijing  Institute of Technology}
  \city{Beijing}
\country{China}
}
\email{lironghuabit@126.com}
\begin{abstract}
Higher-order graph clustering aims to partition the graph using frequently occurring subgraphs (i.e., motifs), instead of the lower-order edges, as the  atomic clustering unit, which has been  recognized as the state-of-the-art solution in ground truth community detection and knowledge discovery. \emph{Motif conductance} is one of the most promising higher-order graph clustering models due to its strong interpretability. However, existing motif conductance based graph clustering algorithms are mainly limited by a seminal  two-stage reweighting computing framework,  needing to enumerate  all motif instances  to obtain an edge-weighted  graph for partitioning. However, such a framework has two-fold vital defects: (1) It can only provide a quadratic bound for the motif with three vertices, and whether there is provable clustering quality for other motifs is still an open question. (2) The enumeration procedure of motif instances  incurs prohibitively high  costs against large motifs or large dense graphs due to combinatorial explosions. Besides, expensive spectral clustering or local graph diffusion on the edge-weighted  graph also makes existing methods unable to handle massive graphs with millions of nodes. To overcome these dilemmas, we propose a \textbf{\underline{P}}rovable and \textbf{\underline{S}}calable  \textbf{\underline{M}}otif \textbf{\underline{C}}onductance algorithm \textit{PSMC}, which has a \emph{fixed} and \emph{motif-independent} approximation ratio for any motif. Specifically, \textit{PSMC} first defines a new vertex metric  \concept \ based on the given motif, which can be computed locally. Then,  it  iteratively deletes the vertex with the smallest \smallconcept \ value very efficiently using  novel dynamic update technologies.  Finally, it outputs the locally optimal result during the above iterative process. To further boost efficiency, we propose several effective bounds to estimate the \smallconcept \ value of each vertex, which can greatly reduce  computational costs.  Empirical results on real-life and synthetic  demonstrate that  our proposed algorithms achieve  3.2$\sim$32 times  speedup  and improve the quality by at least 12 times  than  the state-of-the art baselines.
\end{abstract} 




\begin{CCSXML}
	<ccs2012>
	<concept>
	<concept_id>10002950.10003624.10003633.10010917</concept_id>
	<concept_desc>Mathematics of computing~Graph algorithms</concept_desc>
	<concept_significance>500</concept_significance>
	</concept>
	</ccs2012>
\end{CCSXML}

\ccsdesc[500]{Mathematics of computing~Graph algorithms}
\vspace{-0.3cm}

\keywords{Higher-order Graph Clustering; Motif Conductance}
\maketitle

\section{Introduction} \label{sec:intro}
Graph clustering is a fundamental  problem in machine
learning and enjoys numerous applications, including image segmentation \cite{DBLP:conf/cvpr/ShiM97}, anomaly detection \cite{DBLP:conf/kdd/GuptaGSH12}, parallel computing \cite{DBLP:books/siam/06/DevineBK06}, and graph representation learning \cite{DBLP:conf/nips/YingY0RHL18, DBLP:conf/kdd/ChiangLSLBH19,DBLP:conf/kdd/ZhangHZZ20,DBLP:conf/kdd/HuangZXLZ21}.  Therefore, many traditional graph clustering models have been proposed in the literature, such as  null model based (e.g., modularity \cite{newman_finding_2004}), edge cut based (e.g., ratio cut or normalized cut  \cite{von2007tutorial}), and subgraph cohesiveness based (e.g., $k$-core or $k$-truss \cite{DBLP:conf/icde/ChangQ19}). Informally, these models partition all vertices into several clusters, satisfying the vertices within the same cluster have more edges than the vertices in different clusters \cite{trevisan2017lecture}.

 Nevertheless, these traditional graph clustering models ignore the significant motif connectivity patterns (i.e., small frequently occurring subgraphs), which are regarded as indispensable for modeling and understanding the higher-order organization of complex networks \cite{Milo2002NetworkMS,yaverouglu2014revealing}. Unlike dyadic edges, each motif (involves more than two nodes)  indicates the unique interaction behavior among vertices and represents some particular
 functions. To name a few, the triangle is the stable relationship cornerstone  of social networks \cite{DBLP:journals/corr/KlymkoGK14, DBLP:conf/kdd/SotiropoulosT21}. Cycles hint at some money laundering events in financial markets \cite{kloostra2009system}. Feed-forward loops are basic transcription units in genetic  networks \cite{mangan2003structure}. As a consequence, adopting such mesoscopic level motif as the atomic clustering unit has been  recognized as the state-of-the-art (SOTA) solution in ground truth community detection and knowledge discovery  \cite{DBLP:conf/kdd/SotiropoulosT21,DBLP:journals/tnse/XiaYLLL22, DBLP:conf/cikm/DuvalM22}. Such clustering methods are typically named higher-order graph clustering, which aims at capturing the higher-order community structures with dense motifs instead of edges \cite{benson2016higher}. This paper focuses on higher-order graph clustering  on massive graphs with millions of nodes. Thus, perfect solutions  must be highly scalable  and utility.

Numerous higher-order graph clustering models  have been proposed in the literature \cite{arenas2008motif, DBLP:conf/www/Tsourakakis15a,benson2016higher} (Section \ref{sec:related}). Perhaps, the most representative and effective model is the \emph{motif conductance}  due to its strong 
interpretability and  solid theoretical foundation \cite{benson2016higher} (Section \ref{sec:pro}). To be specific, motif conductance is the variant of conductance (conductance is an edge-based clustering model \cite{DBLP:conf/kdd/GleichS12, DBLP:journals/pvldb/GalhotraBBRJ15,DBLP:conf/nips/ZhangR18}), which indicates the ratio of the number of motif instances going out the cluster to the number of  motif instances within the cluster. As a result, smaller motif conductance implies better higher-order clustering quality \cite{DBLP:conf/kdd/ZhouZYATDH17, DBLP:journals/tkdd/ZhouZYATDH21,benson2016higher, DBLP:conf/www/TsourakakisPM17, DBLP:conf/kdd/YinBLG17}. However, identifying the result with the smallest motif conductance   raises significant challenges due to its NP-hardness \cite{benson2016higher}. Therefore, many approximate or heuristic algorithms have been proposed to either improve the clustering quality or reduce the computational costs.  For example,  the \textit{Science} paper \cite{benson2016higher}  proposed a seminal two-stage reweighting framework.
In the first stage,  the input graph $G$ is transformed into an edge-weighted graph $\mathcal{G}^{\mathbb{M}}$, in which the weight of each edge $e$ is the number of motif instances $e$ participates in.  In the second stage, the traditional  spectral clustering is used to partition $\mathcal{G}^{\mathbb{M}}$. However, such a seminal framework can only obtain provable clustering qualities for the motif consisting of three vertices \cite{benson2016higher}. Thus, whether the motif with four or more vertices (such motifs are more realistic \cite{DBLP:journals/csr/YuFZBXX20, DBLP:journals/csur/RibeiroPSAS21}) has a provable clustering quality  is still an open question. On top of that, the framework has to enumerate all motif instances in advance and computes the  eigenvector of normalized Laplacian matrix of $\mathcal{G}^{\mathbb{M}}$, resulting in prohibitively high time and space costs (Section \ref{sec:existing}).  To improve the efficiency, some local graph diffusion algorithms are proposed to replace the eigenvector calculation with various random walk distributions (e.g., Personalized PageRank and higher-order markov chain) (Section \ref{sec:existing}).  However, these  algorithms are heuristic,  and their clustering qualities are heavily dependent on many hard-to-tune parameters and seeding  strategies.  So, their performance is unstable and in most cases very poor, as demonstrated in our experiments.  Recently, Huang et al.  \cite{DBLP:conf/icde/HuangLBL21}  pointed out that almost all existing solutions are limited by the above two-stage reweighting framework, and then they proposed an adaptive sampling method to estimate the weights of the edges for reducing the computational time. However, this adaptive sampling method introduces randomness, leading to inaccurate results. Therefore, obtaining provable and scalable algorithms for  motif conductance  remains a challenging task.

To overcome the above limitations, we propose a  \textbf{\underline{P}}rovable and \textbf{\underline{S}}calable  \textbf{\underline{M}}otif \textbf{\underline{C}}onductance algorithm \textit{PSMC}. Since the purpose of optimizing motif conductance is to obtain target clusters rather than to obtain the intermediate edge-weighted graph $\mathcal{G}^{\mathbb{M}}$, it is not necessary to blindly spend much time on getting precise $\mathcal{G}^{\mathbb{M}}$. Instead, we deeply analyze the functional form of motif conductance (Lemma \ref{lem:opt_mc}) and  iteratively optimize motif conductance starting from each vertex. Specifically, we first define a new vertex metric  \concept \ (Definition \ref{def:mr}),  which can be computed locally. Then,   \emph{PSMC}  iteratively deletes the vertex with the smallest \smallconcept \ value very efficiently using  novel dynamic update technologies.  Finally, \emph{PSMC} returns the cluster with the smallest motif conductance during the above iterative process.  As a consequence, \emph{PSMC}  is to integrate the computation and partition of edge-weighted graph $\mathcal{G}^{\mathbb{M}}$ in an iterative algorithm, thus eliminating the need for expensive spectral clustering or local graph diffusion. Besides, we also theoretically prove that \emph{PSMC} has a \emph{fixed} and   \emph{motif-independent} approximation ratio (Theorem \ref{thm:1}). In other words, \emph{PSMC} can  output a fixed  approximation ratio for any given motif, which solves the open question posed by the previous two-stage reweighting framework. Particularly, when the given motif has three vertices, \emph{PSMC} improves the well-known quadratic bound (Table \ref{tab:alg}).  On the other hand, the \smallconcept \ of vertex $u$ implicitly depends on the number of motif instances $u$ participates in, causing \emph{PSMC} also indirectly calculates all motif instances.  To further boost efficiency,  we develop several effective bounds to estimate the \smallconcept \ of each vertex via the well-known Turan  Theorem \cite{turan} and colorful $h$-star degree \cite{DBLP:conf/icde/GaoLQCYW22}. We highlight our  main contributions as follows.

\stitle{A Novel Computing Framework with Accuracy Guarantee.} We introduce a provable and scalable  motif conductance algorithm, called \textit{PSMC}, based on the proposed vertex metric  \concept. \emph{PSMC} has  two striking features.  One is that it is a novel high-order graph clustering framework by integrating the computation and partition of edge-weighted graph in an iterative algorithm,  reducing the computational costs. The other is that it  can output a fixed  and motif-independent approximation ratio for any given motif, while the existing SOTA  frameworks cannot.

\stitle{Several Effective Optimization Strategies.} To further boost efficiency, we develop  several dynamic update technologies to incrementally maintain the \smallconcept \ of each vertex when its neighbor is deleted, without recomputing the \smallconcept \ from scratch.  Besides, with the help of Turan  Theorem and  colorful $h$-star degree, several effective  bound estimation strategies are proposed  to obtain a better trade-off between efficiency and accuracy.

\stitle{Extensive Experiments.} We conduct extensive experiments on nine datasets (five real-world graphs and four synthetic graphs) and  eight competitors to evaluate the scalability and effectiveness of our proposed solutions. These empirical results show that our algorithms achieve  3.2$\sim$32 times  speedup  and improve the quality by at least 12 times  than  baselines. Besides, our algorithms  realize up to an order of magnitude memory reduction when contrasted with baselines.

\section{Preliminaries} 
\begin{figure}[t]
	\centering
	{
		\includegraphics[width=0.4\textwidth]{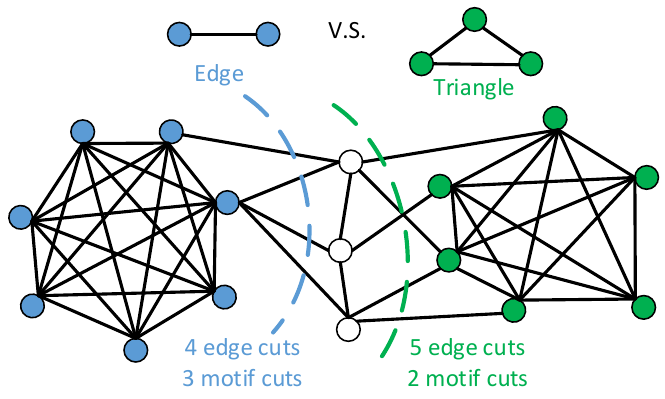}
	} \vspace{-0.3cm}
	\caption{Illustration of the traditional edge-based conductance and the motif conductance on a synthetic graph. There are 47 edges and 60 triangles. The blue dotted line indicates the optimal cut when the motif is an edge and the corresponding conductance is $\frac{4}{\min\{42,52\}}$. The green dotted line represents the optimal cut when the motif is a triangle and the corresponding triangle conductance is $\frac{2}{\min\{116,64\}}$. Motif conductance is more likely to preserve motif instances compared with edge-based conductance.} \vspace{-0.3cm}
	\label{fig:intro}
\end{figure} 
\subsection{Problem Formulation} \label{sec:pro}

Given an unweighted and  undirected graph $G(V,E)$\footnote{For simplicity, we consider the unweighted and  undirected graph in this work, while our proposal can be easily extended to the weighted or directed graphs.}, we use  $V$ and $E$ to represent the vertex set and the edge set of $G$, respectively.   We denote $|V|=n$ (resp.$|E|=m$) as the number of vertices (resp. edges) of $G$.  Let $G_{S}(S,E_S)$ be the induced subgraph induced by $S$ iff $S \subseteq V$ and $E_S=\{(u,v)\in E| u,v \in S\}$. We use $N_{S}(v)=\{u\in S|(u,v)\in E\}$ to denote the neighbors of $v$  in $S$. We use $\mathbb{M}$ to denote the use-initiated query motif, which is a frequently occurring interaction pattern (i.e., significant subgraph) in complex networks. For simplicity,  $G_S \in \mathbb{M}$ means that $G_S$ is an instance of $\mathbb{M}$. Namely, $G_S \in \mathbb{M}$ iff $G_S$  is isomorphic to $\mathbb{M}$\footnote{$G(V_1,E_1)$ and $G(V_2,E_2)$ are isomorphic if there exists a bijection $f$: $V_1 \rightarrow V_2$ such that $(u,v)\in E_2$ iff $(f(u),f(v)) \in E_2$.}. We let $k(\mathbb{M})$ be the order of $\mathbb{M}$, which is the number of vertices involved in  $\mathbb{M}$. For example, an edge is a second-order motif and a triangle is a  three-order motif.  Following existing research work \cite{DBLP:conf/kdd/YinBLG17,  DBLP:conf/www/Tsourakakis15a, DBLP:conf/www/Fu0MCBH23}, unless otherwise stated, we also take clique as the motif by default in this paper. A high-level definition of higher-order graph clustering is as follows.

\begin{definition} [higher-order graph clustering]\label{def:high}
For an unweighted and  undirected graph $G(V,E)$ and a   motif \ $\mathbb{M}$,  the problem of the higher-order graph clustering aims to find a high-quality cluster $C \subseteq V$ has the following properties:
(1) $G_C$ contains many instances of \ $\mathbb{M}$; (2) there are few motif instances that cross $ G_C $ and $ G_{V\setminus C}$.
\end{definition}

Based on the intuition of Definition \ref{def:high}, we use the most representative and effective  \emph{motif conductance} \cite{DBLP:conf/kdd/ZhouZYATDH17, DBLP:journals/tkdd/ZhouZYATDH21,benson2016higher, DBLP:conf/www/TsourakakisPM17, DBLP:conf/kdd/YinBLG17} to measure the clustering quality of an identified cluster $C$.

\begin{definition} [motif conductance] \label{def:mc}
	For an unweighted and  undirected graph $G(V,E)$ and a   motif \ $\mathbb{M}$, the motif conductance of $C$ is defined as	$\phi_{\mathbb{M}}(C)=\frac{cut_{\mathbb{M}}(C)}{\min\{vol_{\mathbb{M}}(C), vol_{\mathbb{M}}(V\setminus C)\}}$. 
	\begin{equation} \label{def:cut}
		cut_{\mathbb{M}}(C)=|\{G_S\in \mathbb{M} | S\cap C \neq \emptyset,  S\cap (V\setminus C) \neq \emptyset\}|
		\end{equation}
		\begin{equation} \label{def:vol}
	vol_{\mathbb{M}}(C)=\sum_{u \in C} |\{G_S\in \mathbb{M} | u \in S\}|
	\end{equation}
	
\end{definition}

\begin{table*}[t!]
	\centering
	\caption{A comparison of motif conductance based  graph clustering. $\phi_{\mathbb{M}}^*\in (0,1]$ is the smallest motif conductance value. $k=k(\mathbb{M})$ is the order of $\mathbb{M}$. $\delta$ is the degeneracy and is often very small in real-life graphs \cite{DBLP:journals/pvldb/LiGQWYY20}.  $t_{max}$ and $b$ are the maximum iteration number and motif volum parameter of $\textit{HOSPLOC}$. $\epsilon$ is the error tolerance of \textit{MAPPR} to execute forward push. \textit{PSMC+} is \textit{PSMC} with  estimate strategies proposed in Section \ref{section:4.2}. $\times$ represents the corresponding method has no accuracy guarantee.}
	\scalebox{1}{
		\begin{tabular}{c|c|c|c|c}
			\toprule
			\multicolumn{1}{c|}{Methods} & \multicolumn{1}{c|}{Accuracy Guarantee}&
			\multicolumn{1}{c|}{Time Complexity}&\multicolumn{1}{c|}{Space Complexity}&
			\multicolumn{1}{c}{Remark}\\
			\midrule
			\multirow{2}{*}	{\textit{HSC} \cite{benson2016higher}} & $O(\sqrt{\phi_{\mathbb{M}}^*})$ for $k=3$ & 	\multirow{2}{*}	{$O(k(\frac{\delta}{2})^{k-2}m)+O(n^3)$}  & $O(n^2)$&	\multirow{2}{*}	{ Eigenvector-based}  \\
			&$\times$ for $k>3$ & & \\
			\midrule
			\textit{HOSPLOC} \cite{DBLP:conf/kdd/ZhouZYATDH17, DBLP:journals/tkdd/ZhouZYATDH21}	 & $\times$ &$O(t_{max}\frac{2^{bk}}{(\phi_{\mathbb{M}}^*)^{2k}}\log^{3k}m )$ &$O(n^k)$ &  Higher-order Markov Chain-based\\
			\textit{MAPPR}  \cite{DBLP:conf/kdd/YinBLG17} & $\times$ & $O(k(\frac{\delta}{2})^{k-2}m)+O(\frac{\log \frac{1}{\epsilon}}{ \epsilon})$ & $O(n^2)$& Personalized PageRank-based\\
			\midrule
			\textit{PSMC} (This paper)& $O(1/2+1/2\phi_{\mathbb{M}}^*)$ for any $k$& $O(k(\frac{\delta}{2})^{k-2}m)$ & $O(m+n)$ & \concept-based\\
			\textit{PSMC+} (This paper) & $\times$ & $O(km)$ & $O(m+n)$ & \concept-based\\
			\bottomrule	
		\end{tabular}
	}\vspace{-0.3cm}
	\label{tab:alg} 
\end{table*}

Where $cut_{\mathbb{M}}(C)$  is the number of motif instance with at least one vertex in $C$ and at least one vertex in $V\setminus C$, and $vol_{\mathbb{M}}(C)$ (resp. $vol_{\mathbb{M}}(V\setminus C)$) is the number of the motif instance the vertices in  $C$ (resp. $V\setminus C$) participate in. When $\mathbb{M}$ is an edge, the motif conductance degenerates into classic conductance \cite{DBLP:conf/kdd/GleichS12, DBLP:journals/pvldb/GalhotraBBRJ15,DBLP:conf/aaai/LinLJ23}. Thus, edges that do not participate in any motif instances do not contribute to the motif conductance. Namely, a cluster with many edges but few motif instances may also have poor motif conductance. Therefore, motif conductance has strong  interpretability and can improve the quality of the resulting cluster by focusing on the particular motifs that are important higher-order structures of a given network  \cite{benson2016higher}. Figure \ref{fig:intro} shows  the difference between  the traditional edge-based conductance and the motif conductance. Note that  we have $\phi_{\mathbb{M}}(C)=\phi_{\mathbb{M}}(V\setminus C)$ by Definition \ref{def:mc}. % Our problem are formulated as follows.

\stitle{Problem Statement.} Given an unweighted and  undirected graph $G(V,E)$ and a   motif  $\mathbb{M}$, the goal of motif conductance based  graph clustering is to find a vertex subset $S^* \subseteq V$, satisfying $vol_{\mathbb{M}}(S^{*})\leq vol_{\mathbb{M}}(V\setminus S^{*})$ and  $\phi_{\mathbb{M}}(S^*)\leq \phi_{\mathbb{M}}(S)$ for any $S \subseteq V$.  $\phi_{\mathbb{M}}^*$ stands for $\phi_{\mathbb{M}}(S^*)$ for brevity.

\subsection{Existing Solutions and Their Shortcomings} \label{sec:existing}

In this subsection, we review several SOTA motif conductance algorithms, which can be roughly divided into two categories: seed-free global clustering and seed-dependent local clustering.

\subsubsection{Seed-free global clustering.} Seed-free global clustering  identifies the higher-order clusters by calculating the eigenvector of the normalized Laplacian matrix of the edge-weighted graph $\mathcal{G}^{\mathbb{M}}$, in which the weight of each edge $e \in \mathcal{G}^{\mathbb{M}}$ is the number of motif instances $e$ participates in. For example, Benson et al. \cite{benson2016higher} proposed the following two-stage higher-order spectral clustering (\emph{HSC}). Specifically, \emph{HSC} first obtains the normalized Laplacian matrix  $\mathcal{L}$ by enumerating motif instances. Then, \emph{HSC} computes the eigenvector $x$ of the second smallest eigenvalue of $\mathcal{L}$ to execute the sweep procedure. Namely, it sorts all entries in $x$ such that $x_1\leq x_2 \leq ... \leq x_n$ and outputs $S=\arg\min \phi(S_i)$, in which $S_i=\{x_1,x_2,...,x_i\}$.  The following theorems are important theoretical basis of \emph{HSC}.

\begin{theorem}  [\cite{benson2016higher}]
	Given a graph $G(V,E)$ and a motif \ $\mathbb{M}$, for any $S \subseteq V$, we have
\begin{equation}
\phi_{\mathbb{M}}(S)=
\begin{cases}
\phi^{\mathcal{G}^\mathbb{M}}(S), &if \ k(\mathbb{M})=3\\
\phi^{\mathcal{G}^\mathbb{M}}(S)-\frac{\sum_{mi \in \mathbb{M}}I(|mi \cap S|=2)}{\sum_{u\in S}\mathcal{D}_{uu}}, &if \ k(\mathbb{M})=4\\
\end{cases}
\end{equation}
Where $\phi^{\mathcal{G}^\mathbb{M}}(S)$ is the edge-based  conductance of  $S$ in terms of the weighted graph $\mathcal{G}^\mathbb{M}$ and $I(.)$ is the indicator function. Note that when $k(\mathbb{M})>4$,  the relationship between $\phi_{\mathbb{M}}(S)$ and $\phi^{\mathcal{G}^\mathbb{M}}(S)$ is unclear.
\end{theorem}

\begin{theorem} [Cheeger inequality \cite{benson2016higher}]
Given a graph $G(V,E)$ and a motif \ $\mathbb{M}$ with $k(\mathbb{M})=3$,	let  $S$ be the vertex subset returned by \textit{HSC}, we have $\phi_{\mathbb{M}}^*\leq \phi_{\mathbb{M}}(S) \leq 2\sqrt{\phi_{\mathbb{M}}^*}$, in which $\phi_{\mathbb{M}}^*$ is the optimal  motif conductance.
\end{theorem}

\vspace{-0.2cm}
\stitle{Discussions.} Note that \emph{HSC} can only derive the Cheeger inequality for motifs consisting of three vertices. However, it has not been proven whether there is quality guarantee for motifs with four or more vertices. On top of that, since \emph{HSC} needs to enumerate motif instances in advance and calculate the eigenvector of $\mathcal{L}$, its time complexity is $O(k(\frac{\delta}{2})^{k-2}m)+O(n^3)$ and space complexity is $O(n^2)$ \cite{benson2016higher}, resulting in poor scalability.

%\vspace{-0.2cm}
\subsubsection{Seed-dependent local clustering} \label{subsec:gr}
Seed-dependent local clustering executes the local graph diffusion from the given seed vertex $q$ to identify higher-order clusters. Higher-order markov chain based random walk \cite{rankone14} and  Personalized PageRank based random walk  \cite{DBLP:conf/kdd/0001YXWY17, DBLP:conf/sigmod/WeiHX0SW18} are two well-known local graph diffusion methods. The former  uses state transition tensors  to simulate the long-term dependence of states. The latter models a random walk with restart over the  edge-weighted graph $\mathcal{G}^{\mathbb{M}}$.  Based on these backgrounds, Zhou et al.	\cite{DBLP:conf/kdd/ZhouZYATDH17, DBLP:journals/tkdd/ZhouZYATDH21} and Yin et al.	  \cite{DBLP:conf/kdd/YinBLG17} proposed \textit{HOSPLOC}  and \textit{MAPPR} to obtain higher-order clusters, respectively. To be specific, they first compute the probability distribution $\pi$ at the end of the corresponding graph diffusion (i.e.,truncated higher-order markov random walk or Personalized PageRank random walk), and let $y=\pi \mathcal{D}^{-1}$. Then, they run the sweep procedure. Namely, it sorts all non-zero entries in $y$ such that $y_1\geq y_2 \geq ... \geq y_{sup(y)}$  ($sup(y)$ is the number of the non-zero entries in $y$), and outputs $S=\arg\min \phi(S_i)$, in which $S_i=\{y_1,y_2,...,y_i\}$.  

\stitle{Discussions.} Since seed-dependent local clustering methods aim to identify the higher-order clusters to which the given seed vertex $q$ belongs, they only have locally-biased Cheeger-like quality for the motif consisting of three vertices \cite{DBLP:conf/kdd/ZhouZYATDH17, DBLP:journals/tkdd/ZhouZYATDH21,DBLP:conf/kdd/YinBLG17}. Namely, seed-dependent local clustering methods do not give the theoretical gap to $\phi_{\mathbb{M}}^*$. Besides, their clustering qualities are heavily dependent on many hard-to-tune parameters and seeding  strategies. Practically, their performance is unstable and even find degenerate solutions, as demonstrated in our experiments.

\vspace{-0.2cm}
\subsubsection{The Shortcomings of Existing Solutions} 
Table \ref{tab:alg} summarizes the above SOTA motif conductance algorithms for comparison. By Table \ref{tab:alg}, we know that the complexities of our solutions
are lower than the baselines. This is because baselines
need to enumerate all motif instances in advance, and then
execute the expensive spectral clustering or local graph diffusion.
However, we are to integrate the enumeration and partitioning in an iterative algorithm, which can greatly reduce  computational costs. On top of that, our \emph{PSMC} can output $O(1/2+1/2\phi_{\mathbb{M}}^*)$ accuracy guarantee for any size of motif, while baselines cannot. Note that although the proposed \emph{PSMC+}  has no accuracy guarantee (the practical performance of \emph{PSMC+} is comparable to the baselines in our empirical results), it has the excellent property of near-linear time complexity, which is vital for dealing with massive graphs with millions of nodes.

\section{PSMC: The Proposed Solution} \label{sec:our}

In this section, we devise  a  \textbf{\underline{P}}rovable and \textbf{\underline{S}}calable  \textbf{\underline{M}}otif \textbf{\underline{C}}onductance algorithm \textit{PSMC}, which aims to output a high-quality cluster. It is important to highlight that PSMC has the capability to provide \emph{fixed} and \emph{motif-independent} approximation ratio for any motif. This significant feature addresses and resolves the open problem raised by \cite{benson2016higher}. Then, we propose novel dynamic update technologies and effective  bounds to  further boost efficiency of \textit{PSMC}.

\subsection{The \textit{PSMC} Algorithm}

Recall that our problem is to obtain a  higher-order cluster rather than to obtain the intermediate edge-weighted graph $\mathcal{G}^{\mathbb{M}}$, thus it is not necessary to blindly spend much time on getting precise $\mathcal{G}^{\mathbb{M}}$. Based on in-depth observations, we reformulate  motif conductance and propose a novel computing framework, which iteratively optimized motif conductance starting from each vertex. Before describing our proposed algorithms, several useful definitions are stated as follows. 
\begin{definition} [Motif Degree]
Given an unweighted and undirected graph $G(V,E)$ and a motif \ $\mathbb{M}$ with $k=k(\mathbb{M})$, the  motif degree of $u$ is defined as $\mathbb{M}(u)=|\{G_S \in \mathbb{M}|u\in S\}|$. For a positive integer $1\leq j \leq k$,  we let  $\mathbb{M}_{j}^{C}(u)=|\{G_S \in \mathbb{M}|u\in S, |C\cap S|=j\}|$. 
\end{definition}

\begin{definition} [\smallconcept] \label{def:mr}
	Given an unweighted and undirected graph $G(V,E)$, a motif \ $\mathbb{M}$ with $k=k(\mathbb{M})$, and a vertex subset $S$, the \smallconcept \ of $u$ $\in S$ w.r.t. $G_S$ is defined as $Mr_S(u)=\frac{\mathbb{M}(u)+\mathbb{M}_{k}^{S}(u)-\mathbb{M}_{1}^{S}(u)}{\mathbb{M}(u)}$.
\end{definition}

Based on these definitions, we  develop  \textit{PSMC} with three-stage computing framework  (Algorithm \ref{alg:framework}). In \textit{Stage 1}, we  compute the  \smallconcept \ value for each vertex (Lines 1-7).  In \textit{Stage 2}, we iteratively remove the vertex with the smallest \smallconcept \ (Lines 8-11). Such an iterative deletion process is referred to as a \emph{peeling} process.   In \textit{Stage 3}, we output the result with the smallest motif conductance during the peeling process (lines 12-13). Next, we prove that this simple \textit{PSMC} algorithm can produce   the high-quality cluster with \emph{fixed} and  \emph{motif-independent} approximation ratio.

\begin{lemma}[Monotonicity] \label{lem:monotony}
	Given an unweighted and undirected graph $G(V,E)$, a motif \ $\mathbb{M}$ with $k=k(\mathbb{M})$, and two vertex subset $S$ and $H$, we have $\mathbb{M}_{k}^{H}(u)\geq\mathbb{M}_{k}^{S}(u)$ and $\mathbb{M}_{1}^{H}(u)\leq \mathbb{M}_{1}^{S}(u)$ if $u \in S\subseteq H$.
\end{lemma}

\begin{lemma} \label{lem:cut}
	Given an unweighted and undirected graph $G(V,E)$, a motif \ $\mathbb{M}$ with $k=k(\mathbb{M})$, and a vertex subset $S$, we have  $cut_{\mathbb{M}}(S\setminus \{u\})=cut_{\mathbb{M}}(S)-\mathbb{M}_{1}^{S}(u)+\mathbb{M}_{k}^{S}(u)$.
\end{lemma}

Let $g_\mathbb{M}(S)=(\sum\limits_{u\in S}\mathbb{M}(u)-cut_{\mathbb{M}}(S))/(\sum\limits_{u\in S}\mathbb{M}(u))$ and assume that the larger of $g_\mathbb{M}(S)$, the better the quality of $S$. Let $\widetilde{S}$ be the optimal vertex set for $g_\mathbb{M}(.)$. That is, $g_\mathbb{M}(\widetilde{S})\geq g_\mathbb{M}(S)$ for any vertex subset $S\subseteq V$. The following two lemmas are key theoretical basis of \textit{PSMC}.

\begin{lemma} [Reformulation of Motif Conductance] \label{lem:opt_mc}
	Given an unweighted and undirected graph $G(V,E)$, a motif \ $\mathbb{M}$ with $k=k(\mathbb{M})$, and a vertex subset $S$, we have $\phi_{\mathbb{M}}(S)=1-g_\mathbb{M}(S)$ if $vol_{\mathbb{M}}(S)\leq vol_{\mathbb{M}}(V\setminus S)$.
\end{lemma}

\begin{lemma} \label{lem:opt_g}
	Given an unweighted and undirected graph $G(V,E)$, a motif \ $\mathbb{M}$ with $k=k(\mathbb{M})$, we have  $Mr_{\widetilde{S}}(u)\geq g_\mathbb{M}(\widetilde{S})$ for any $u \in \widetilde{S}$. 
\end{lemma}

\stitle{Implications of Lemma \ref{lem:opt_mc} and Lemma \ref{lem:opt_g}.} Since $g_\mathbb{M}(\widetilde{S})\geq g_\mathbb{M}(S)$ for any $S\subseteq V$,  Lemma \ref{lem:opt_mc} indicates that $\phi_{\mathbb{M}}(\widetilde{S})=\phi_{\mathbb{M}}(S^{*})$ where $S^{*}$ is  our optimal vertex set.  This is because that $S^{*}$ satisfies $vol_{\mathbb{M}}(S^{*})\leq vol_{\mathbb{M}}(V\setminus S^{*})$, thus the condition of $vol_{\mathbb{M}}(S)\leq vol_{\mathbb{M}}(V\setminus S)$ in Lemma \ref{lem:opt_mc} is always true for our problem. Please see Problem Statement in Section \ref{sec:pro} for details. Meanwhile,   Lemma \ref{lem:opt_g} indicates that  the \smallconcept \ of any vertex $u  \in \widetilde{S}$ w.r.t $\widetilde{S}$  is at least $g_\mathbb{M}(\widetilde{S})$. Namely,  the \smallconcept \ of any vertex in $S^{*}$  w.r.t $S^{*}$  is at least $1-\phi^{*}_{\mathbb{M}}$.  Based on these implications, we can derive the following theorem to local a cluster with \emph{fixed} and  \emph{motif-independent} approximation ratio.

\begin{algorithm}[t]%\scriptsize
	\caption{\textbf{\underline{P}}rovable and \textbf{\underline{S}}calable  \textbf{\underline{M}}otif \textbf{\underline{C}}onductance  (\textit{PSMC})}
	\label{alg:framework}
	\begin{algorithmic}[noline]
	\State \textbf{Input}: 	A graph $G(V,E)$ and a motif $\mathbb{M}$ with $k=k(\mathbb{M})$
	\State \textbf{Output}: A higher-order cluster $\hat{S}$ with fixed and motif-independent approximation ratio
	\end{algorithmic}
\begin{algorithmic}[1] %[1] enables line numbers
\State  $MI \leftarrow$ all the motif instances of $\mathbb{M}$
\State Initializing the motif degree  $\mathbb{M}(u)=0$ for any $u\in V$
\For {each motif instance $mi \in MI$}
\For {each edge $(u,v) \in mi$}
\State  $\mathbb{M}(u)+=1$; $\mathbb{M}(v)+=1$
\EndFor
\EndFor 
		\State $i\leftarrow 1$; $S_i\leftarrow V$; $\mathbb{M}_{k}^{S_i}(u)=\mathbb{M}(u)$ and $\mathbb{M}_{1}^{S_i}(u)=0$ for $u \in S_i$
		\State 	$Mr_{S_i}(u)\leftarrow \frac{\mathbb{M}(u)+\mathbb{M}_{k}^{S_i}(u)-\mathbb{M}_{1}^{S_i}(u)}{\mathbb{M}(u)}$ for any $u \in S_i$

		\While{$S_i\neq \emptyset$}
				\State $u \leftarrow \arg \min\{Mr_{S_i}(u)|u\in S_i\}$
		\State $i\leftarrow i+1$

		\State $S_i\leftarrow S_{i-1} \setminus \{u\}$
		\EndWhile
		\State $\hat{S} \leftarrow \mathop{\arg\min}\limits_{S \in \{S_1,S_2,...,S_{n}\}} \{\phi_{\mathbb{M}}(S)|vol_{\mathbb{M}}(S)\leq vol_{\mathbb{M}}(V\setminus S)\}$
		\State \textbf{return}  $\hat{S}$

	\end{algorithmic}
\end{algorithm}

\begin{theorem} \label{thm:1}
	Algorithm \ref{alg:framework} can identify a higher-order cluster with motif conductance $1/2+1/2\phi_\mathbb{M}^*$. 
\end{theorem}

\begin{proof}
According to the definitions of  $cut_{\mathbb{M}}(C)$, $\mathbb{M}(u)$, and $\mathbb{M}_{j}^{C}(u)$, we know that $\mathbb{M}(u)=\sum\limits_{j=1}^{k}\mathbb{M}_{j}^{C}(u)$ and $cut_{\mathbb{M}}(C)=\sum\limits_{u \in C}\sum\limits_{j=1}^{k-1}\frac{1}{j}\mathbb{M}_{j}^{C}(u)$. 	Let $\widetilde{S}$ is the optimal vertex set for $g_\mathbb{M}(.)$. In Lines 8-11, Algorithm \ref{alg:framework} executes the  peeling process. That is, in each round, it greedily deletes the vertex with the smallest \smallconcept. Consider the round $t$ when the first vertex $v$ of $\widetilde{S}$ is deleted. Let $V_t$ be the vertex set from the beginning of round $t$.  $\widetilde{S}$ is the subset of $V_t$ because $v$ is the first deleted vertex of $\widetilde{S}$. This implies that
	$\min\limits_{u \in V_t}Mr_{V_t}(u)=Mr_{V_t}(v)=\frac{\mathbb{M}(v)+\mathbb{M}_{k}^{V_t}(v)-\mathbb{M}_{1}^{V_t}(v)}{\mathbb{M}(v)}\geq \frac{\mathbb{M}(v)+\mathbb{M}_{k}^{\widetilde{S}}(v)-\mathbb{M}_{1}^{\widetilde{S}}(v)}{\mathbb{M}(v)}=Mr_{\widetilde{S}}(v)\geq g_\mathbb{M}(\widetilde{S})$ according to Lemma \ref{lem:opt_g} and Lemma \ref{lem:monotony}.  Therefore, for any $u\in V_t$, we have $\frac{\mathbb{M}(u)+\mathbb{M}_{k}^{V_t}(u)-\mathbb{M}_{1}^{V_t}(u)}{\mathbb{M}(u)}\geq g_\mathbb{M}(\widetilde{S})$. Furthermore, 
	\vspace{-0.2cm}
		\begin{footnotesize}
	\begin{align}
	&g_\mathbb{M}(V_t)=\frac{\sum\limits_{u\in V_t}\mathbb{M}(u)-cut_{\mathbb{M}}(V_t)}{\sum\limits_{u\in V_t}\mathbb{M}(u)}=\frac{\sum\limits_{u\in V_t}(\mathbb{M}(u)-\sum\limits_{j=1}^{k-1}\frac{1}{j}\mathbb{M}_{j}^{V_t}(u))}{\sum\limits_{u\in V_t}\mathbb{M}(u)}\\
	&=\frac{\sum\limits_{u\in V_t}(\sum\limits_{j=2}^{k-1}(1-\frac{1}{j})\mathbb{M}_{j}^{V_t}(u)+\mathbb{M}_{k}^{V_t}(u))}{\sum\limits_{u\in V_t}\mathbb{M}(u)}\geq \frac{\sum\limits_{u\in V_t}(\sum\limits_{j=2}^{k-1}\frac{1}{2}\mathbb{M}_{j}^{V_t}(u)+\mathbb{M}_{k}^{V_t}(u))}{\sum\limits_{u\in V_t}\mathbb{M}(u)}\\
	&=\frac{\sum\limits_{u\in V_t}(\sum\limits_{j=2}^{k-1}\mathbb{M}_{j}^{V_t}(u)+2\cdot\mathbb{M}_{k}^{V_t}(u))}{2\sum\limits_{u\in V_t}\mathbb{M}(u)}=\frac{\sum\limits_{u\in V_t}(\mathbb{M}(u)+\mathbb{M}_{k}^{V_t}(u)-\mathbb{M}_{1}^{V_t}(u))}{2\sum\limits_{u\in V_t}\mathbb{M}(u)}\\
	&\geq \frac{1}{2}\frac{\sum\limits_{u\in V_t}g_\mathbb{M}(\widetilde{S})\cdot \mathbb{M}(u)}{\sum\limits_{u\in V_t}\mathbb{M}(u)}=\frac{1}{2}g_\mathbb{M}(\widetilde{S}).  
	\end{align}
		\end{footnotesize}

	Since  Algorithm \ref{alg:framework} maintains the optimal solution during the peeling  process in Lines 12-13, $\phi_{\mathbb{M}}(\hat{S})=1-g_\mathbb{M}(\hat{S})\leq 1- g_\mathbb{M}(V_t)\leq  1-\frac{g_\mathbb{M}(\widetilde{S})}{2}$ due to Lemma \ref{lem:opt_mc}. On the other hand,  According to the definition of $\widetilde{S}$, we know that $g_\mathbb{M}(\widetilde{S})\geq g_\mathbb{M}(S^*)$, in which $S^{*}$ is the vertex set with optimal motif conductance. Thus, $\phi_\mathbb{M}(\hat{S})\leq 1-\frac{g_\mathbb{M}(\widetilde{S})}{2}\leq 1-\frac{g_\mathbb{M}(S^*)}{2}=1-\frac{1-\phi_\mathbb{M}(S^*)}{2}$. Namely, $\phi_\mathbb{M}(\hat{S})\leq 1/2+1/2\phi_\mathbb{M}(S^*)$. As a result, Algorithm \ref{alg:framework} can identify a higher-order  cluster with motif conductance  $1/2+1/2\phi_\mathbb{M}^*$.
	\end{proof}

\subsection{Dynamic  Update of \concept}
The  computational challenge of Algorithm \ref{alg:framework} is how to incrementally  maintain $Mr_{S_{i}}$ in Line 9 and $\phi_{\mathbb{M}}(S_{i})$ in Line 12  when a vertex $u$ is removed. Note that since $\phi_{\mathbb{M}}(S_{i})=1-g_{\mathbb{M}}(S_{i})$ by Lemma \ref{lem:opt_mc}, we can maintain $\phi_{\mathbb{M}}(S_{i})$ by maintaining $g_{\mathbb{M}}(S_{i})$. We propose efficient dynamic update technologies to solve the challenge as follows. 

\begin{lemma}\label{lem:update}
	For the current search space $S_i$, if  a vertex $u \in S_i$ is removed, for any $v \in N_{S_i}(u)$ and $S_{i+1}=S_{i}\setminus \{u\}$, we have the following equation:
	\vspace{-0.3cm}
	\begin{equation} \label{eq:mr}
		Mr_{S_{i+1}}(v)=Mr_{S_i}(v)-\frac{\mathbb{M}_{k}^{S_i}(u,v)+\mathbb{M}_{2}^{S_i}(u,v)}{\mathbb{M}(u)}
	\end{equation}
	\begin{equation} \label{eq:g}
		g_{\mathbb{M}}(S_{i+1})
		=\frac{vol_{\mathbb{M}}(S_{i})g_{\mathbb{M}}(S_i)-Mr_{S_i}(u)\mathbb{M}(u)}{vol_{\mathbb{M}}(S_{i})-\mathbb{M}(u)}
	\end{equation}
\end{lemma}
Where $\mathbb{M}_{k}^{S_i}(u,v)=|\{G_S \in \mathbb{M}|(u,v) \in G_S, |S_i\cap S|=k\}|$ and $\mathbb{M}_{2}^{S_i}(u,v)=|\{G_S \in \mathbb{M}|(u,v) \in G_S, |S_i\cap S|=2\}|$. That is $\mathbb{M}_{k}^{S_i}(u,v)$  (resp., $\mathbb{M}_{2}^{S_i}(u,v)$) is  the number of motif instances containing the edge $(u,v)$ with exactly $k$ (resp., 2) vertices in $S_i$.

Based on Lemma \ref{lem:update}, Algorithm \ref{alg:framework} can incrementally update the \smallconcept \ for each vertex when its neighbor is removed. The time complexity of Algorithm \ref{alg:framework}  is analyzed as follows.

\begin{theorem} \label{thm:time}
	The  worse-case time complexity of Algorithm \ref{alg:framework} is $O(k(\delta/2)^{k-2}m)$. $k$ is the order of $\mathbb{M}$. $\delta$ is the degeneracy and is often very small in real-world graphs (Table \ref{tab:data}). 
\end{theorem}

\subsection{Lower and Upper Bounds of \concept} \label{section:4.2}

According to the definitions of $Mr_S(.)$ and $g_{\mathbb{M}}(.)$, we can know that the initial $Mr_{S_1}(u)=2$ for any vertex $u \in S_1$ and $g_{\mathbb{M}}(S_1)=1$. Thus, by Equation \ref{eq:mr} and Equation \ref{eq:g}, we further know that the bottleneck of Algorithm \ref{alg:framework} is how to quickly obtain the motif degree $\mathbb{M}(u)$, $\mathbb{M}_{k}^{S_i}(u,v)$ and $\mathbb{M}_{2}^{S_i}(u,v)$. Inspired by this, we propose effective lower and upper bounds to estimate them,  which can be computed locally. Specifically, let $NS_{u}$, $NS^{S}_{uv}$, and $NS^{V\setminus S}_{uv}$ be the \textit{neighbor subgraph} induced by $N(u)$, $N_S(u)\cap N_{S}(v)$, and $N_{V\setminus S}(u)\cap N_{V\setminus S}(v)$, respectively. We have  $\mathbb{M}(u)$ and $\mathbb{M}_{k}^{S_i}(u,v)$(resp., $\mathbb{M}_{2}^{S_i}(u,v)$) are the number of  $(k-1)$-cliques  in $NS_{u}$ and  the number of  $(k-2)$-cliques  in $NS^{S_i}_{uv}$ (resp., $NS^{V\setminus S_i}_{uv}$), respectively. Therefore, the estimate of $\mathbb{M}(u)$, $\mathbb{M}_{k}^{S_i}(u,v)$ and $\mathbb{M}_{2}^{S_i}(u,v)$ becomes to estimate  the number of  cliques in the corresponding subgraph.

\subsubsection{Lower Bounds.} For convenience, we use $NS$ to specify which  neighbor subgraph is adopted, i.e.,   $NS_{u}$, $NS^{S_{i}}_{uv}$, and $NS^{V\setminus S_{i}}_{uv}$.
The following Theorem is one of the most important results in extremal graph theory, which can be used to estimate the number of cliques.

%Note that the lower (resp., upper) bound of $\mathbb{M}(u)$ can directly deduce the upper (resp., lower) bound of $Mr_S(u)$ due to  $Mr_S(u)=\frac{\mathbb{M}(u)+\mathbb{M}_{k}^{S}(u)-\mathbb{M}_{1}^{S}(u)}{\mathbb{M}(u)}$ (Definition \ref{def:mr}). Meanwhile, we can also use the estimation bounds of $\mathbb{M}(u)$ to  initialize $Mr_S(u)$. Namely, we use the estimated $\mathbb{M}(u)$ to replace the calculations in Lines 1-6 of Algorithm \ref{alg:framework}. 

\begin{theorem} [Turan Theorem  \cite{turan} \label{thm:turan}]
For any subgraph $NS$, if $\frac{2E(NS)}{|V(NS)|(|V(NS)|-1)}> 1-\frac{1}{r-1}$, then $NS$ contains a $r$-clique.
\end{theorem}

According to Theorem \ref{thm:turan}, we have the following facts.

\begin{fact} \label{fact:1}
Let $D=\frac{2E(NS)}{|V(NS)|(|V(NS)|-1)}$ and $r=\lfloor \frac{1}{1-D} \rfloor +1$, we have:	(1) $\mathbb{M}(u)\geq \binom{r}{k-1}$ if $NS=NS_{u}$; (2) $\mathbb{M}_{k}^{S_i}(u,v)\geq \binom{r}{k-2}$ if $NS=NS^{S_{i}}_{uv}$; (3) $\mathbb{M}_{2}^{S_i}(u,v)\geq \binom{r}{k-2}$ if $NS=NS^{V\setminus S_{i}}_{uv}$.
\end{fact}

The well-known graph theory expert \textit{Paul Erdos} proposed the following tighter theorem to further expand the Turan Theorem.

\begin{theorem}  \cite{erdos} \label{thm:erdos}
	For any subgraph $NS$, if $\frac{2E(NS)}{|V(NS)|(|V(NS)|-1)}> 1-\frac{1}{r-1}$, then $NS$ contains at least    $(\frac{|V(NS)|}{r-1})^{r-2}$ $r$-cliques.
\end{theorem}

Let $h$ is an integer and $A_{i}=\{C_i^{1},C_{i}^{2},...,C_{i}^{\binom{r}{h}}\}$ be the $h$-clique set obtained from the $i$-th $r$-clique of Theorem \ref{thm:erdos}, in which $i \in \{1,2,...,(\frac{|V(NS)|}{r-1})^{r-2}\}$.  Since two $r$-cliques have at most $r-1$ common vertices, $|A_i \cap A_j| \leq \binom{r-1}{h}$.  According to the inclusion-exclusion principle \cite{DBLP:journals/siamcomp/BjorklundHK09} and let $t=(\frac{|V(NS)|}{r-1})^{r-2}$, we have $|\bigcup\limits_{i=1}^{t} A_i | \geq \sum\limits_{i=1}^{t} |A_i|-\sum\limits_{1\leq i<j \leq t}|A_i\cap A_i| \geq t\binom{r}{h}-\binom{t}{2}\binom{r-1}{h}$. So, we have the following facts.

\begin{fact} \label{fact:2}
Let $D=\frac{2E(NS)}{|V(NS)|(|V(NS)|-1)}$, $r=\lfloor \frac{1}{1-D} \rfloor +1$, and $t=(\frac{|V(NS)|}{r-1})^{r-2}$, we have: (1)	$\mathbb{M}(u)\geq t\binom{r}{k-1}-\binom{t}{2}\binom{r-1}{k-1}$  if $NS=NS_{u}$; (2)	$\mathbb{M}_{k}^{S_i}(u,v)\geq t\binom{r}{k-2}-\binom{t}{2}\binom{r-1}{k-2}$ if $NS=NS^{S_{i}}_{uv}$; (3)	$\mathbb{M}_{2}^{S_i}(u,v)\geq t\binom{r}{k-2}-\binom{t}{2}\binom{r-1}{k-2}$ if $NS=NS^{V\setminus S_{i}}_{uv}$.
\end{fact}

In a nutshell, we can obtain the lower bounds of $\mathbb{M}(u)$, $\mathbb{M}_{k}^{S_i}(u,v)$ and $\mathbb{M}_{2}^{S_i}(u,v)$ according to Fact \ref{fact:1} and Fact \ref{fact:2}.

\begin{figure}[t]
	\centering
	{
		\includegraphics[width=0.4\textwidth]{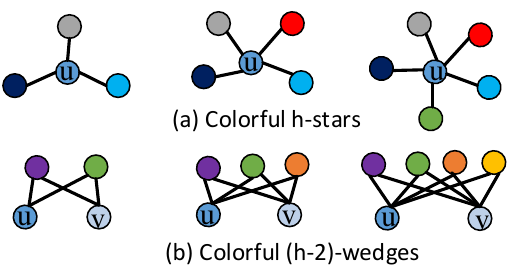}
	} \vspace{-0.3cm}
	\caption{Colorful $h$-stars and colorful ($h-2$)-wedges.} \vspace{-0.3cm}
	\label{fig:upper}
\end{figure}

\subsubsection{Upper Bounds.} Gao et al. proposed  the concept of colorful $h$-star degree $csd_h(u)$, which can be  computed in $O(h|N(u)|)$ time. Specifically, $csd_{h}(u)$ is the number of colorful $h$-stars centered on $u$, in which a colorful $h$-star is a star with $h$ vertices having different colors (Figure \ref{fig:upper} (a)). Since each vertex in $h$-clique  must have a different color, we can obtain $csd_{h}(u)\geq \mathbb{M}(u)$ if $h=k(\mathbb{M})$. For estimating $\mathbb{M}_{k}^{S_i}(u,v)$ and $\mathbb{M}_{2}^{S_i}(u,v)$,  we propose a novel concept of colorful $(h-2)$-wedge degree  $cwd_h^{S}(u,v)$ w.r.t. the vertex set $S$.   $cwd_h^{S}(u,v)$  is the number of colorful $(h-2)$-wedges of $S$ with $u$ and $v$  as endpoints (there may be no edge between $u$ and $v$), in which a colorful $(h-2)$-wedge of $S$ is  a set of $h-2$ wedges (a wedge is a path with three nodes) such that all vertices are in $S$ and having different colors (Figure \ref{fig:upper} (b)). Assume that $(u,v)$ is an edge and  $hc(u,v)$ is any $h$-clique containing the edge $(u,v)$, we know that each vertex in $hc(u,v)$  must have a different color. Thus,   $cwd_{h}^{S_i}(u,v)\geq\mathbb{M}_{k}^{S_i}(u,v)$  and $cwd_{h}^{V \setminus S_i}(u,v)\geq  \mathbb{M}_{2}^{S_i}(u,v)$ if $h=k(\mathbb{M})$.

So far we have obtained the upper and lower bounds of $\mathbb{M}(u)$, $\mathbb{M}_{k}^{S_i}(u,v)$ and $\mathbb{M}_{2}^{S_i}(u,v)$, which can be computed and updated in near-linear time by the dynamic programming.  Therefore, we directly take the average of their upper and lower bounds as the corresponding estimated value.

\section{Experimental Evaluation}\label{sec:experiments}  
%In this section, we conduct extensive experiments to answer the following Research Questions. \textbf{RQ1: } Can our proposed \emph{PSMC} identify higher quality clusters than the SOTA traditional/higher-order graph clustering methods? \textbf{RQ2: } How much our proposed \emph{PSMC} improve in terms of running time or quality compared to other existing motif conductance algorithms? 

\subsection{Experimental Setup}
\stitle{Datasets.} Our solutions are evaluated on  five real-world graphs with ground-truth clusters\footnote{All datasets can be downloaded from  http://snap.stanford.edu/}, which are widely used benchmarks for higher-order graph clustering \cite{benson2016higher,  DBLP:journals/pvldb/FangYCLL19, DBLP:conf/aaai/HuangCX20}.  Besides, we also use four types of synthetic graphs  to test the scalability and the effectiveness of our solutions: \textit{LFR} \cite{lancichinetti2009detecting},  \textit{PLC} \cite{holme2002growing}, \textit{ER} \cite{erdos1960evolution}, and \textit{BA} \cite{barabasi1999emergence}, which can be generated by the well-known NetworkX Python package \cite{SciPyProceedings_11}. Note that they can be used to simulate the degree distributions, communities, and small-world properties in the real world.

\stitle{Competitors.}  The following SOTA  baselines are implemented for comparison. (1) Traditional graph clustering: \textit{SC} \cite{alon1985lambda1}, \textit{Louvain}  \cite{blondel2008fast} and \textit{KCore} \cite{kcore_init}; (2) Cohesive subgraph based higher-order graph clustering: \textit{HD} \cite{DBLP:journals/pvldb/FangYCLL19, DBLP:journals/pvldb/SunDCS20, DBLP:journals/pacmmod/HeW00023}; (3) Modularity based higher-order graph clustering: \textit{HM} \cite{ arenas2008motif, DBLP:conf/aaai/HuangCX20}; (4) Motif conductance based higher-order graph clustering:  \textit{HSC} \cite{benson2016higher},  \textit{MAPPR} \cite{DBLP:conf/kdd/YinBLG17},  and \textit{HOSPLOC} \cite{DBLP:conf/kdd/ZhouZYATDH17, DBLP:journals/tkdd/ZhouZYATDH21}.    \textit{PSMC} is our proposed Algorithm \ref{alg:framework} and \textit{PSMC+} is \textit{PSMC} with bound estimation strategies in Section \ref{section:4.2}.

\begin{table}[t]
	\centering
	\caption{Dataset statistics. $\delta$ is the degeneracy.}
	\scalebox{1}{
		\begin{tabular}{ccccccc}
			\toprule
			Dataset & $|V|$  & $|E|$  &  $\delta$ & Description\\
			\midrule
			Amazon & 334,863  & 925,872 &6 & Co-purchase\\
			DBLP & 317,080  & 1,049,866 &113 &Collaboration\\
			Youtube & 1,134,890  & 2,987,624 &51 & Social network\\
			LiveJ & 3,997,962  &34,681,189  &360 & Social network\\   
			Orkut & 3,072,441  & 117,185,083 &253 & Social network\\
			\textit{LFR} \cite{lancichinetti2009detecting} & $10^{3} \sim 10^{7}$& $10^{3} \sim 10^{7}$ & $3 \sim 5$ & Synthetic network\\
			\textit{PLC} \cite{holme2002growing}& $10^{3} \sim 10^{7}$&$10^{3} \sim 10^{7}$& $3 \sim 5$& Synthetic network\\
			\textit{ER} \cite{erdos1960evolution}& $10^{3} \sim 10^{7}$& $10^{3} \sim 10^{7}$&5$\sim 11$& Synthetic network\\
			\textit{BA} \cite{barabasi1999emergence}& $10^{3} \sim 10^{7}$& $10^{3} \sim 10^{7}$& $3 \sim 5$ & Synthetic network\\
			\bottomrule		
	\end{tabular}}
	\label{tab:data}
\end{table}

\stitle{Parameters and Implementations.} Unless specified otherwise, we take the default parameters of these baselines in our experiments. Since both \textit{HOSPLOC} and \textit{MAPPR}  take a seed vertex as input, to be more reliable,  we randomly select 50 vertices as seed vertices and report the average runtime and quality. Following previous work  \cite{DBLP:conf/kdd/YinBLG17,  DBLP:conf/www/Tsourakakis15a, DBLP:conf/www/Fu0MCBH23},  we also limit ourselves to the
representative $k$-clique motif to illustrate the main patterns observed. All experiments are conducted  on a Linux server with an Intel(R) Xeon(R) E5-2683 v3@2.00GHZ CPU and 256GB RAM running CentOS 6.10.

\begin{figure*}[t!]
	\centering
		\subfigure[Amazon]{
		\includegraphics[width=0.19\textwidth]{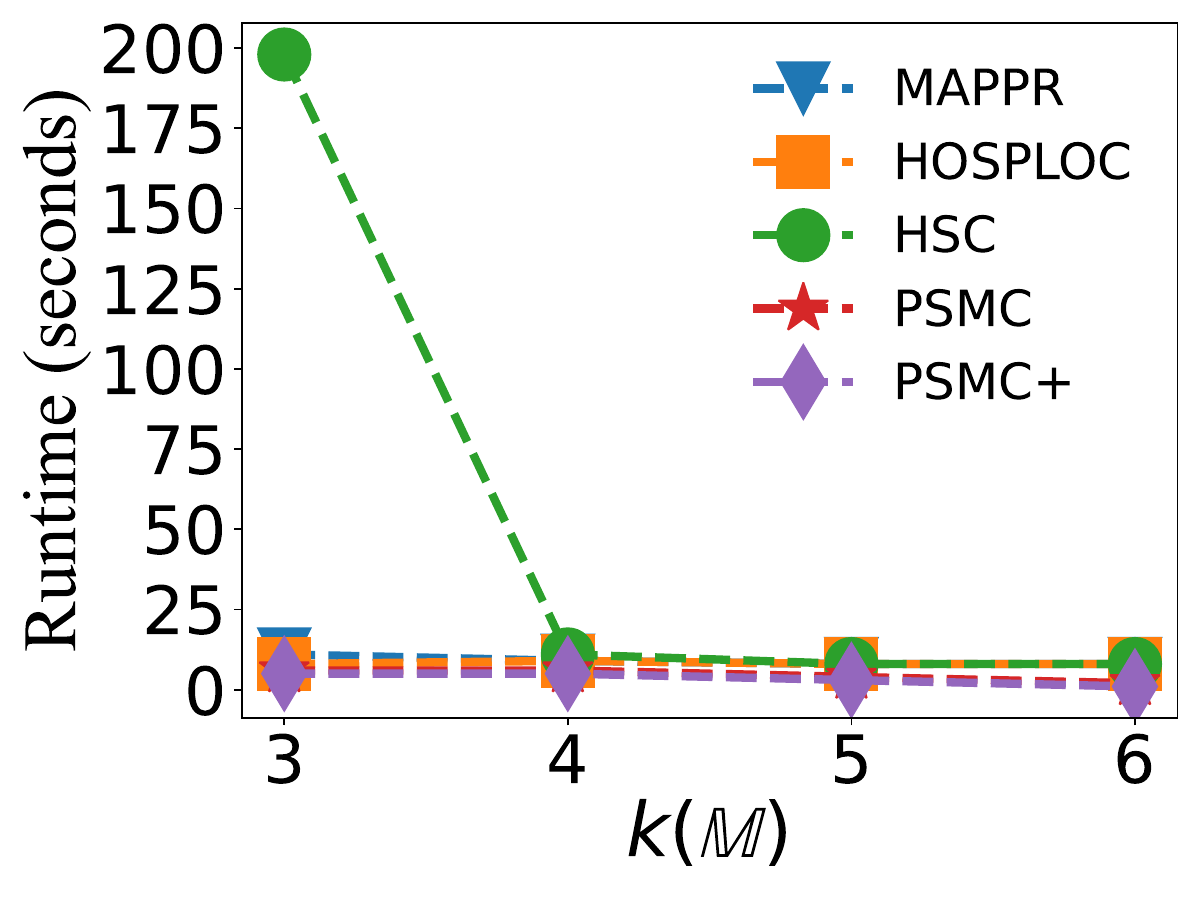}
		\label{fig:time(a)}
	} \hspace{-0.1cm}
	\subfigure[DBLP]{
		\includegraphics[width=0.19\textwidth]{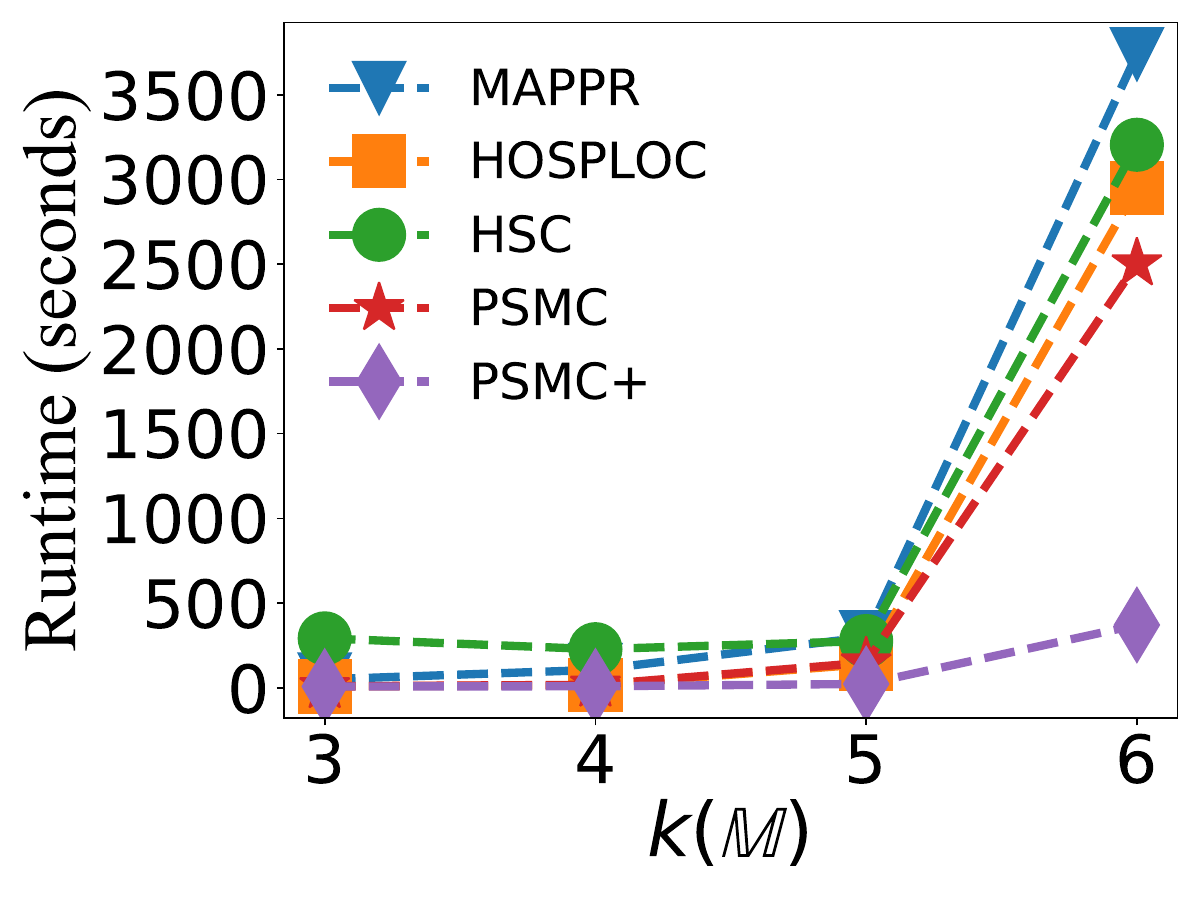}
		\label{fig:time(b)}
	}\hspace{-0.1cm}
	\subfigure[Youtube]{
		\includegraphics[width=0.19\textwidth]{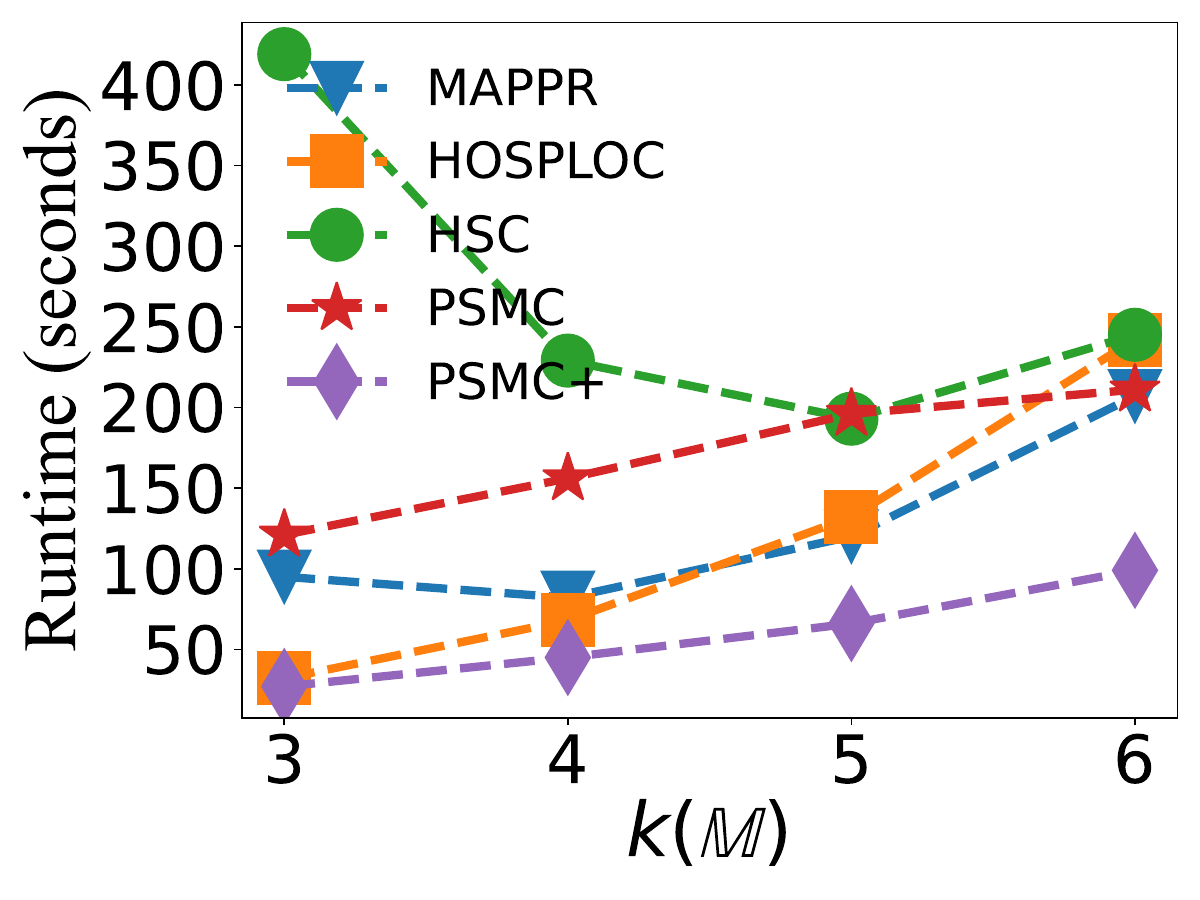}
		\label{fig:time(c)}
	}\hspace{-0.1cm}
	\subfigure[LiveJ]{
		\includegraphics[width=0.19\textwidth]{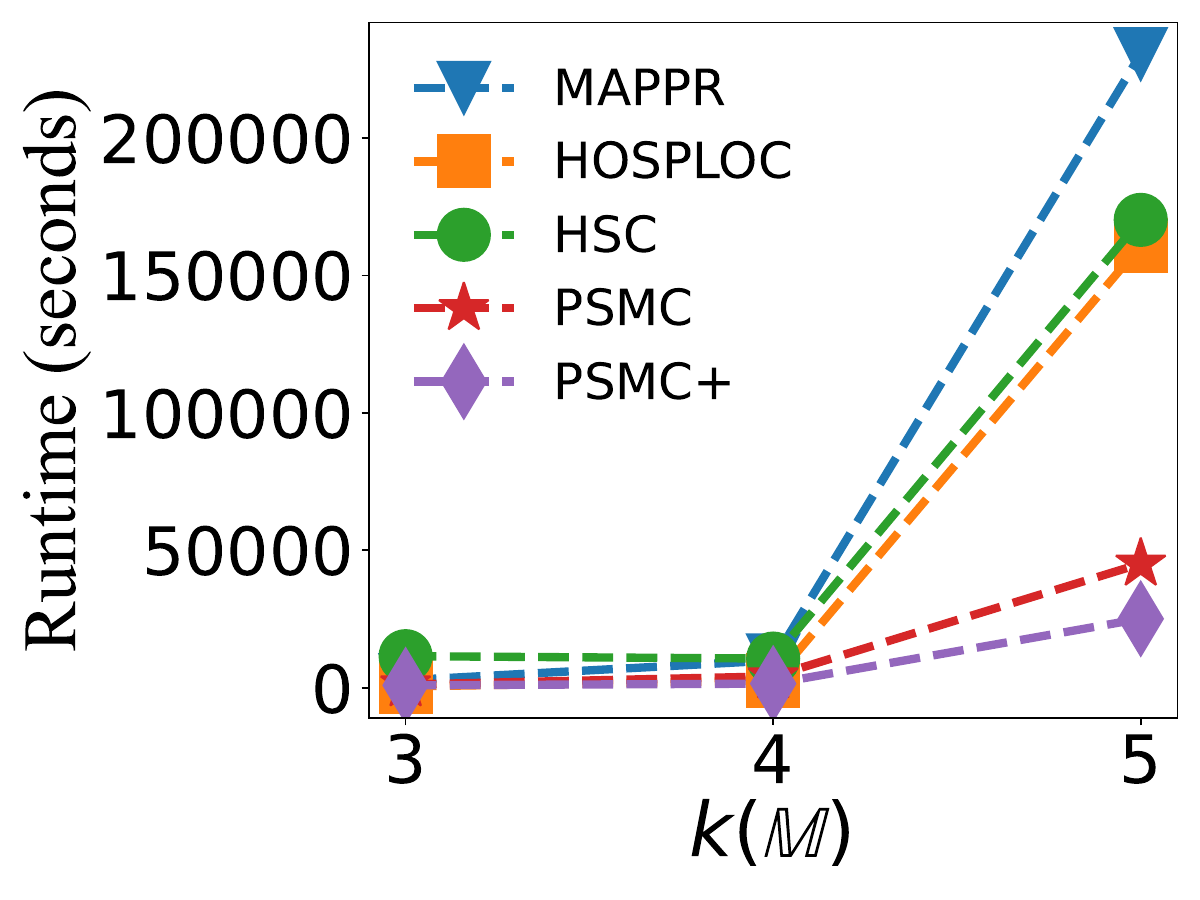}
		\label{fig:time(d)}
	}\hspace{-0.1cm}
	\subfigure[Orkut]{
		\includegraphics[width=0.19\textwidth]{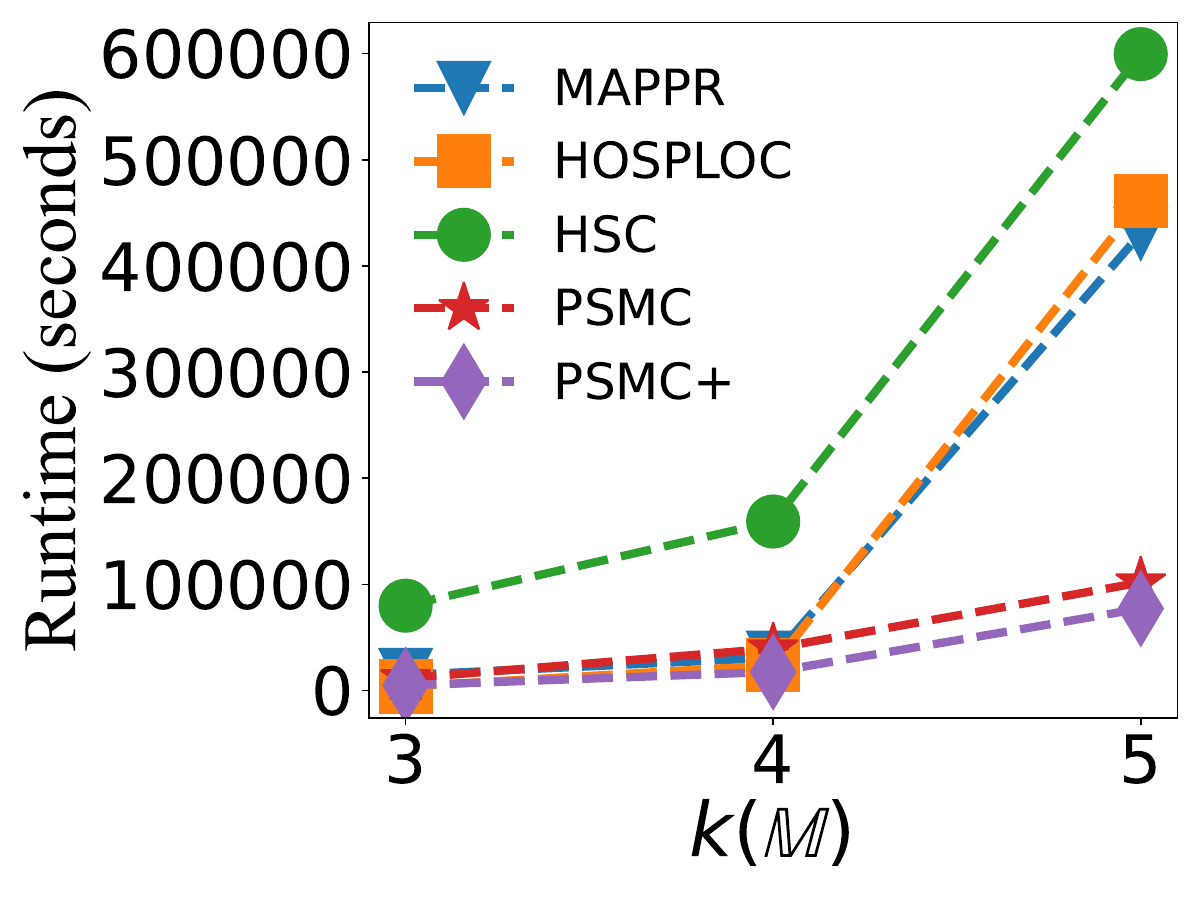}
		\label{fig:time(e)}
	}
	\caption{Runtime (seconds) of  different motif conductance algorithms with varying  $k(\mathbb{M})$.}
	\label{fig:time_world}
\end{figure*}

\begin{figure}[t!]
	\centering
	\subfigure[\textit{LFR}]{
		\includegraphics[width=0.19\textwidth]{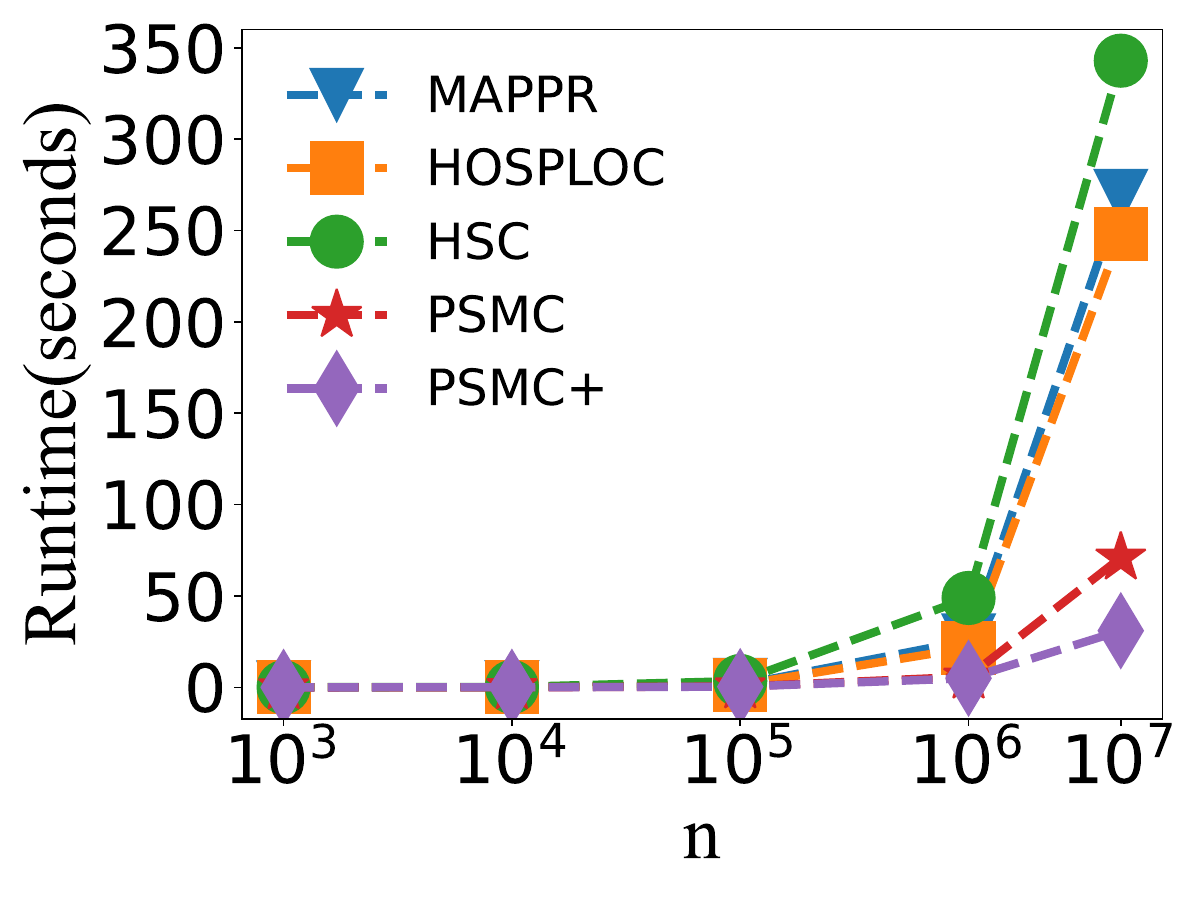}
		\label{fig:time_syn(a)}
	}
	\subfigure[\textit{PLC}]{
		\includegraphics[width=0.19\textwidth]{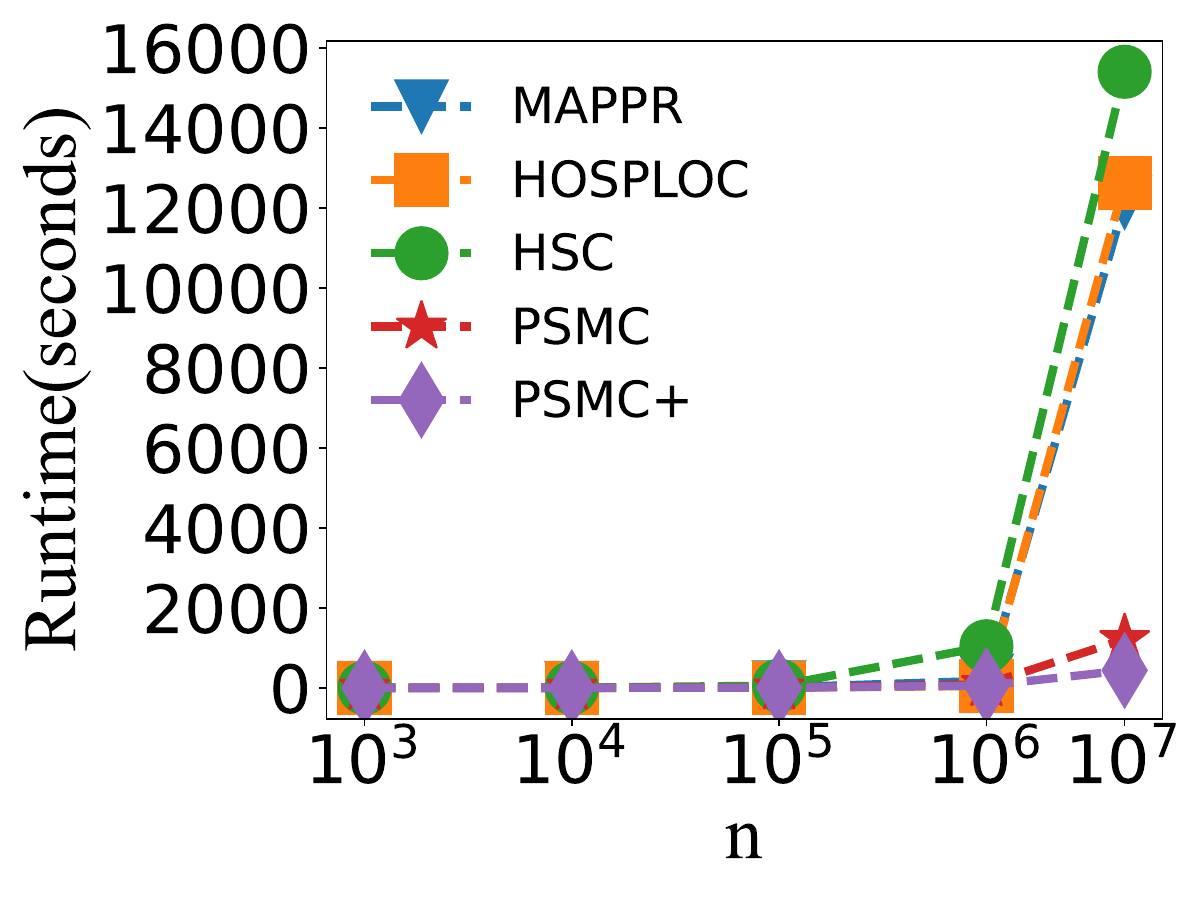}
		\label{fig:time_syn(b)}
	}
	\subfigure[\textit{ER}]{
		\includegraphics[width=0.19\textwidth]{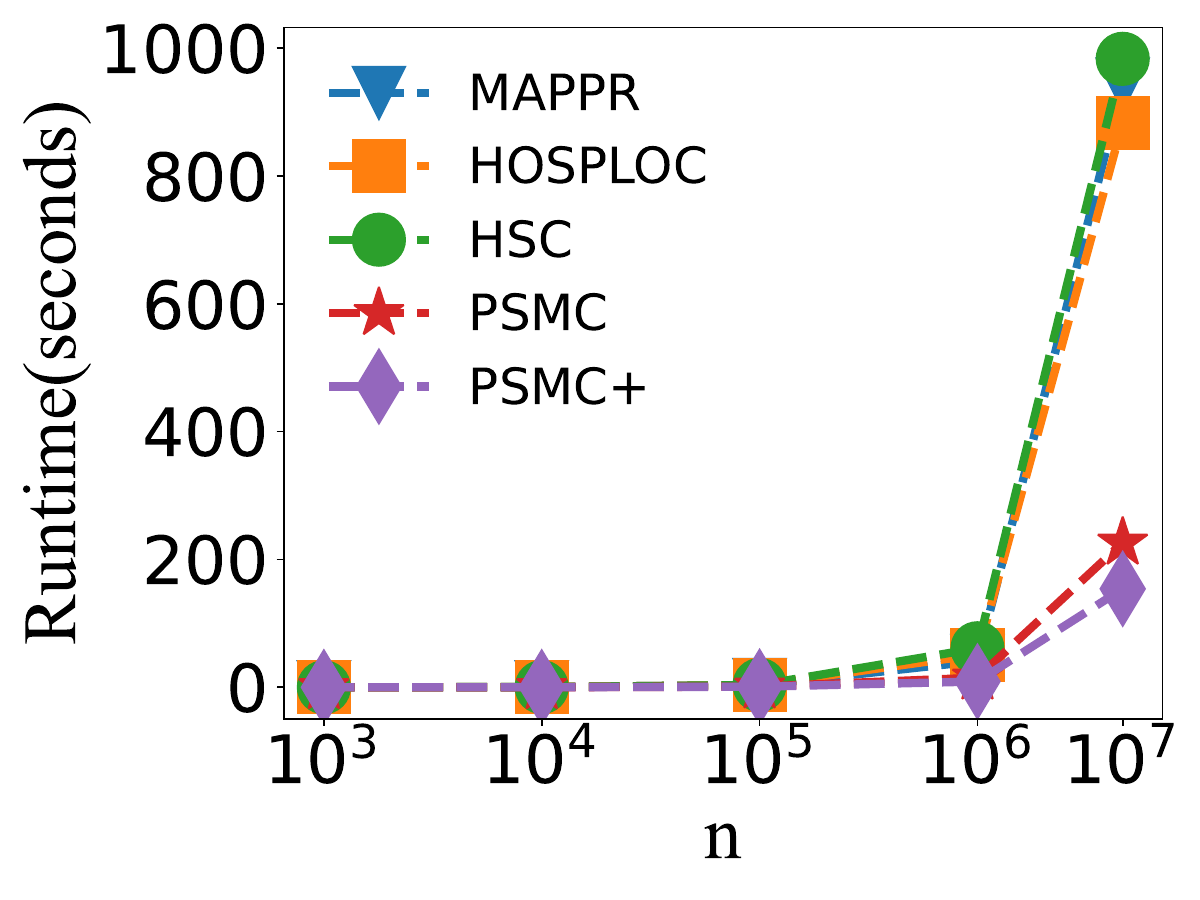}
		\label{fig:time_syn(c)}
	}
	\subfigure[\textit{BA}]{
		\includegraphics[width=0.19\textwidth]{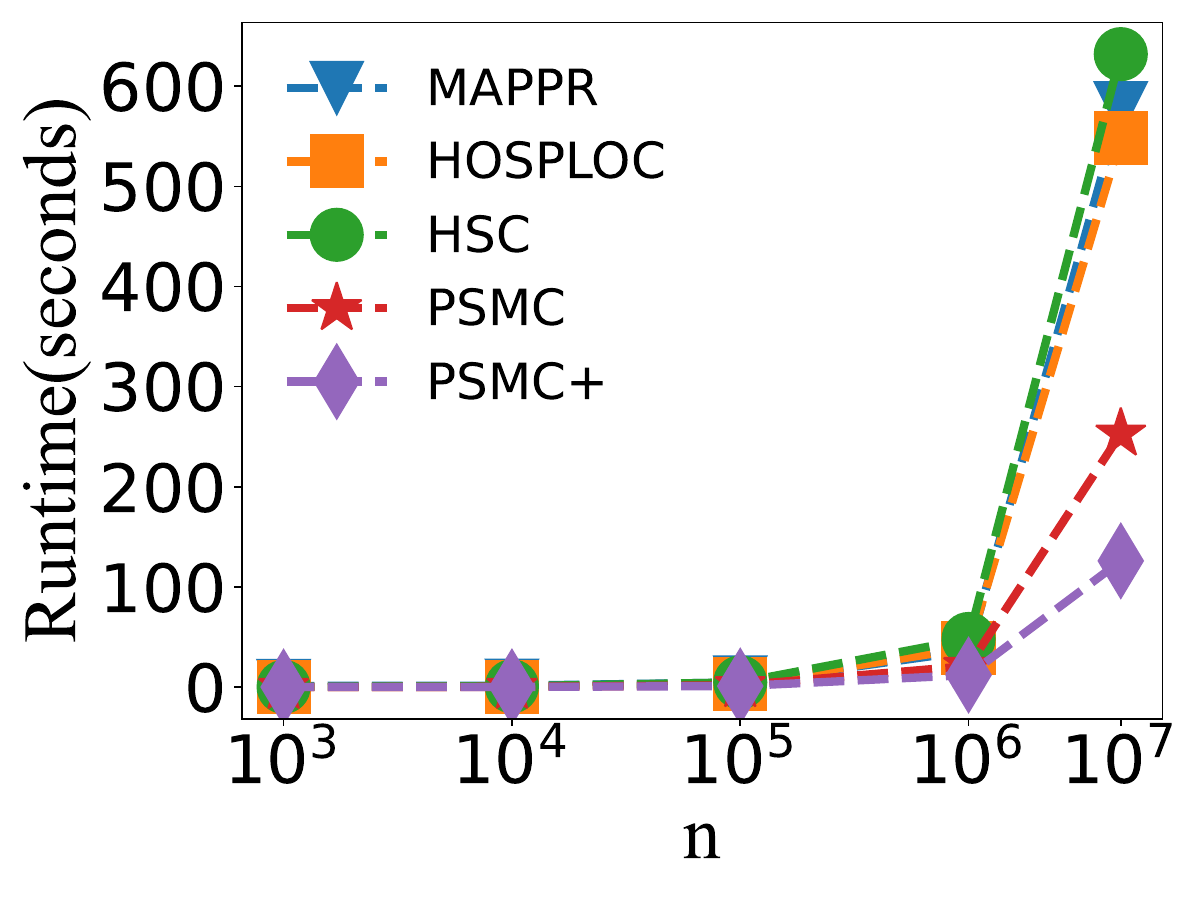}
		\label{fig:time_syn(d)}
	}
	\caption{Scalability testing on synthetic graphs.} 	
	\label{fig:time_syn}
\end{figure} 

\begin{figure}[t!]
	\centering
	\includegraphics[width=0.45\textwidth]{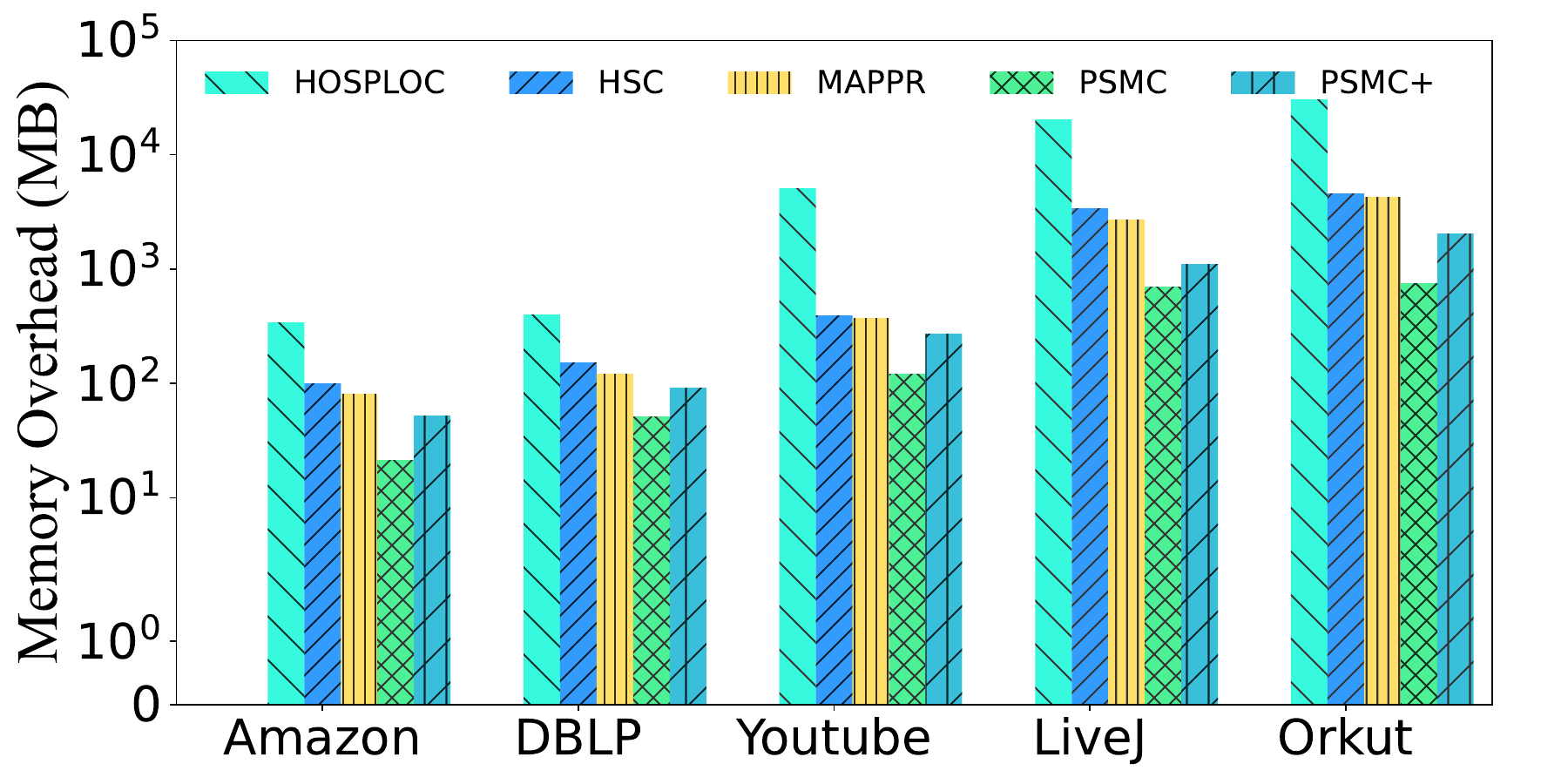}
	\caption{Memory overhead on real-world graphs (excluding the size of the graph itself).}	\vspace{-0.3cm}
	\label{fig:memory} 
\end{figure}

\subsection{Efficiency Testing}
Since the objective functions of traditional graph clustering  (e.g.,
\textit{SC} \cite{alon1985lambda1}, \textit{Louvain}  \cite{blondel2008fast} and \textit{KCore} \cite{kcore_init}),   cohesive subgraph based higher-order graph clustering  (e.g., \textit{HD} \cite{DBLP:journals/pvldb/FangYCLL19, DBLP:journals/pvldb/SunDCS20,DBLP:journals/pacmmod/HeW00023}), and  modularity based higher-order graph clustering  (e.g., \textit{HM} \cite{ arenas2008motif, DBLP:conf/aaai/HuangCX20})  are different from the motif conductance studied in this work, it is meaningless and unnecessary to compare their efficiency.

\stitle{Exp-1: Runtime of different motif conductance algorithms with varying $k(\mathbb{M})$.} The
runtime of \textit{HSC}, \textit{MAPPR}, \textit{HOSPLOC}, \textit{PSMC }and \textit{PSMC+} with varying $k(\mathbb{M})$ on five real-world networks is detailed in Figure \ref{fig:time_world}.  Note that we do not report the empirical results  for  LiveJ and Orkut on $k(\mathbb{M})=6$. This is  because we cannot obtain the results of baselines  (i.e., \textit{MAPPR}, \textit{HOSPLOC}, and \textit{HSC}) within 7 days. By Figure \ref{fig:time_world}, we have: (1)  \textit{PSMC+} is consistently faster than other methods. This is because  \textit{PSMC+} non-trivially adopts  the technologies of Turan  Theorem, colorful $h$-star degree, and colorful $(h-2)$-wedge degree to estimate the \smallconcept \ with  near-linear time (Table \ref{tab:alg}). (2)  \textit{PSMC} is the runner-up on four of the five networks.  The efficiency of  \textit{PSMC} can be attributed to its novel computing framework (i.e., integrating the enumeration and partitioning in an iterative algorithm). However, \textit{MAPPR}, \textit{HOSPLOC}, and \textit{HSC} depend on the weight graph $\mathcal{G}^{\mathbb{M}}$ obtained by enumerating the motif instances, which increases exponentially with the size of the motif (Table 1). In particular, \textit{PSMC}  achieves the speedups of 3.2$\sim$32 times over \textit{HSC}. For example, on DBLP and $k(\mathbb{M})=3$, \textit{PSMC} takes 9 seconds  to obtain the result, while \textit{HSC} takes 292 seconds.  (3) The runtime of all  methods increases with increasing $k(\mathbb{M})$ except for \textit{HSC} on Amazon and Youtube. This is because when $k(\mathbb{M})$  increases, we need  more time to count/estimate motif instances. However, for \textit{HSC}, a possible explanation is that the weighted graph $\mathcal{G}^{\mathbb{M}}$ gets  smaller as $k(\mathbb{M})$  increases, resulting in very little time spent in the spectral clustering stage  of the two-stage reweighting method \cite{benson2016higher}. (4) The sparser the graph (i.e., the smaller the $\delta$), the
faster our algorithms (i.e., \emph{PSMC} and \emph{PSMC+}). For example, our
algorithms  are faster on Youtube compared to DBLP,
despite Youtube having more vertices and more edges (Table
\ref{tab:data}). This is because Youtube has a smaller $\delta$ (Table \ref{tab:data}).  These
results give  preliminary evidence that  the proposed  solutions are indeed high efficiency in practice.

\stitle	{Exp-2: Scalability testing on synthetic graphs.} Extensive synthetic graphs are generated to further test the scalability of our  solutions. Figure \ref{fig:time_syn} only presents the results when the given motif is a triangle,  with
comparable trends across other motifs. By Figure \ref{fig:time_syn}, we know that \textit{PSMC} and \textit{PSMC+} scale near-linear with respect to the graph size. However, the runtime of other  baselines fluctuates greatly as the graph size increases. This is because their time complexity is  nonlinear depending on $n$,  or even is $n^3$  for \textit{HSC} (Table 1). These results indicate that our  algorithms have excellent scalability  over massive graphs while the baselines do not. 

\stitle{Exp-3: Memory overhead comparison.} 
Figure \ref{fig:memory} displays the memory overhead of the evaluated algorithms when the motif is a triangle,  with
comparable trends across other motifs. As excepted, our porposed \textit{PSMC} and \textit{PSMC+} are consistently less than other baselines on all datasets and  realize up to an order of magnitude memory reduction on most cases. This  advantage can be attributed to their novel computing framework, which only  need to maintain some simple data structures, such as \smallconcept  \ of each vertex (Algorithm \ref{alg:framework}). Note that \textit{PSMC+} is slightly worse than \textit{PSMC}. This is because \textit{PSMC+} needs to maintain more intermediate variables to estimate \smallconcept \ (Section \ref{section:4.2}). On the other hand, \textit{HSC} and \textit{MAPPR} exhibit comparable memory overheads, while \textit{HOSPLOC}  has the worst performance. This is because \textit{HOSPLOC} requires storing the expensive state transition tensor (the worst space  is $O(n^3)$) to calculate the vertex ordering required by the sweep procedure. Similarly, \textit{HSC} and \textit{MAPPR} need to store the edge-weighted graph $\mathcal{G}^{\mathbb{M}}$ to calculate this vertex ordering (Section \ref{sec:existing}). These results affirm the memory efficiency of our algorithms.

\begin{table*}[t!]
	\centering
	\caption{Effectiveness of various graph clustering methods. The best and second-best results in each metric are marked in \textbf{bold} and \underline{underlined}, respectively. Note that there is no clear evidence to suggest whether a larger or smaller community size is better. We are simply presenting the community size objectively to provide an intuitive experience.}	\vspace{-0.3cm}
	\scalebox{0.7}{
		\begin{tabular}{cccccccccccccccc}
			\toprule
			\multirow{2}{*}{Model} & \multicolumn{3}{c}{Amazon} & \multicolumn{3}{c}{DBLP} & \multicolumn{3}{c}{Youtube} & \multicolumn{3}{c}{LiveJ}& \multicolumn{3}{c}{Orkut}\\
			\cmidrule(r){2-4}\cmidrule(r){5-7}\cmidrule(r){8-10}\cmidrule(r){11-13}\cmidrule(r){14-16}
			& MC & F1-Score &Size & MC & F1-Score &Size & MC & F1-Score &Size & MC & F1-Score &Size & MC & F1-Score &Size\\
			\midrule
			SC & 0.704 & 0.226& 80793& 0.467& 0.103& 47891 & 0.773& 0.061&38849 & 0.315& 0.306&408010 & 0.499& 0.020 &476537\\
			Louvain & \underline{0.007}& 0.431&239 &0.071& 0.230& 232& 0.467& 0.013 &7480 &0.071& 0.277& 2412 & \underline{0.074}& 0.225 &33\\
			KCore & 0.109 & 0.138& 497& 0.018 & 0.273&114 & 0.470& 0.095 &845 &0.035& 0.313 &377 &  0.141& 0.165 &15706\\
			\midrule
			HD & 0.269& 0.182 &30852 & \underline{0.013 }& 0.242& 309&0.419& 0.074 &1239 & 0.011& 0.125 & 7700& 0.128 & 0.076 &114187\\
			HM & \underline{0.007}& \underline{0.494}&528 &  0.048 & 0.301&237 &0.146& 0.022 &1999 & 0.016& 0.120 &674 & 0.082& \underline{0.248} &96\\
			\midrule
			HSC & 0.067& 0.488&10 & 0.055 & 0.239&427 & 0.074& 0.102& 18& \underline{0.002}& \underline{0.358} & 93 & 0.125& 0.215 &6\\
			MAPPR & 0.015 & 0.457 &175 & 0.115& 0.339&28796 &  0.132& 0.116 &15810 & 0.102& 0.257 &307740 & 0.104& 0.233 &1343943\\
			HOSPLOC & 0.062& 0.467 &90 & 0.260& 0.283& 414&  0.103&0.128 &434 & 0.266& 0.342 &3586 & 0.381& 0.237 & 44651\\
			\midrule
			PSMC & $\mathbf{4*10^{-4}}$& \textbf{0.511} & 74991& $\mathbf{7*10^{-12}}$&\textbf{0.382} & 141& \underline{0.006}& \textbf{0.202}&21342 & $\mathbf{1*10^{-4}}$& \textbf{0.413} & 12458&\textbf{ 0.003} & \textbf{0.312} &1368793 \\
			PSMC+ & 0.012& 0.317&138098 &  0.064& \underline{0.353}&87846 & \textbf{0.000} & \underline{0.138}&229147 & 0.097& 0.326 &88385 &0.483 &0.232 &231334\\
			\bottomrule	
		\end{tabular}
	}	\vspace{-0.3cm}
	\label{table:metric}
\end{table*}

\begin{figure*}[t!]
	\centering
	\subfigure[Amazon]{
		\includegraphics[width=0.19\textwidth]{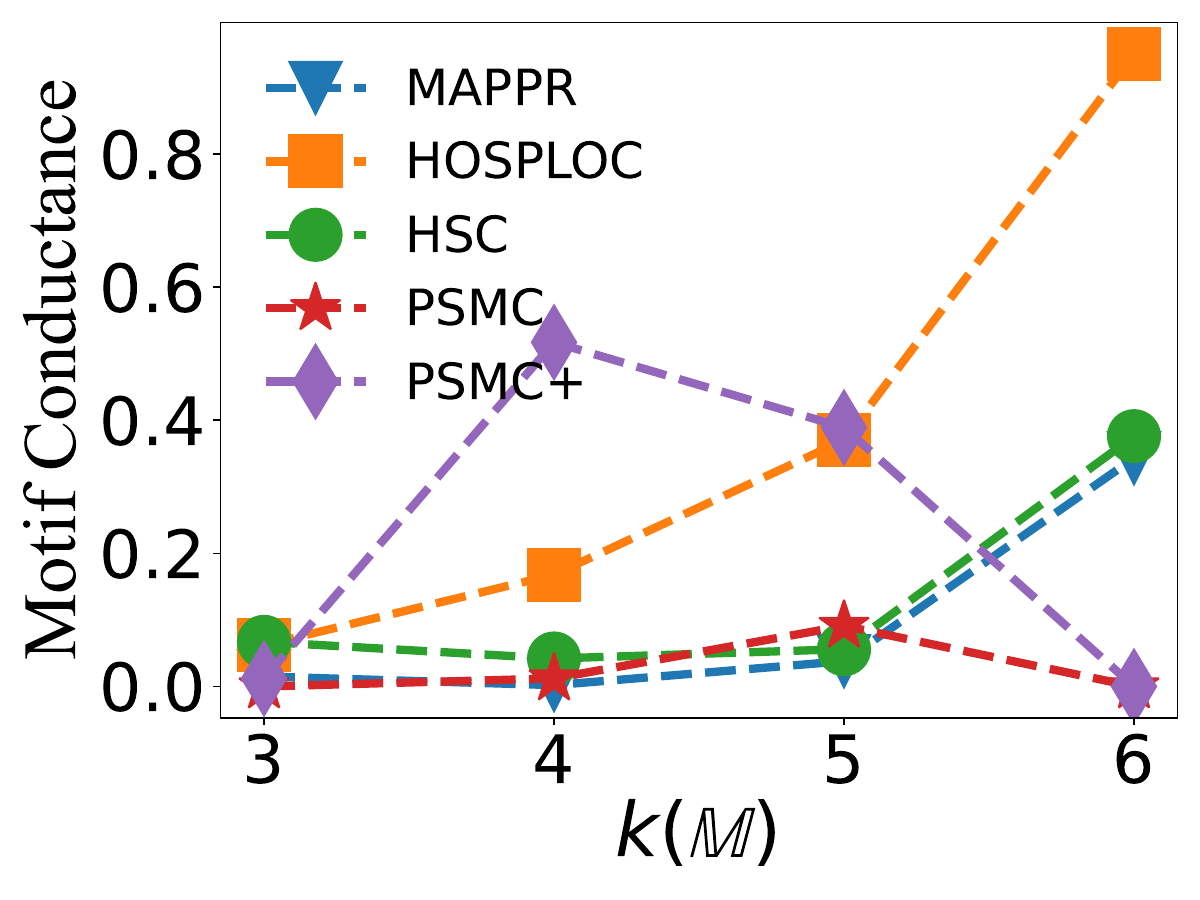}
		\label{fig:mc(a)}
	} \hspace{-0.1cm}
	\subfigure[DBLP]{
		\includegraphics[width=0.19\textwidth]{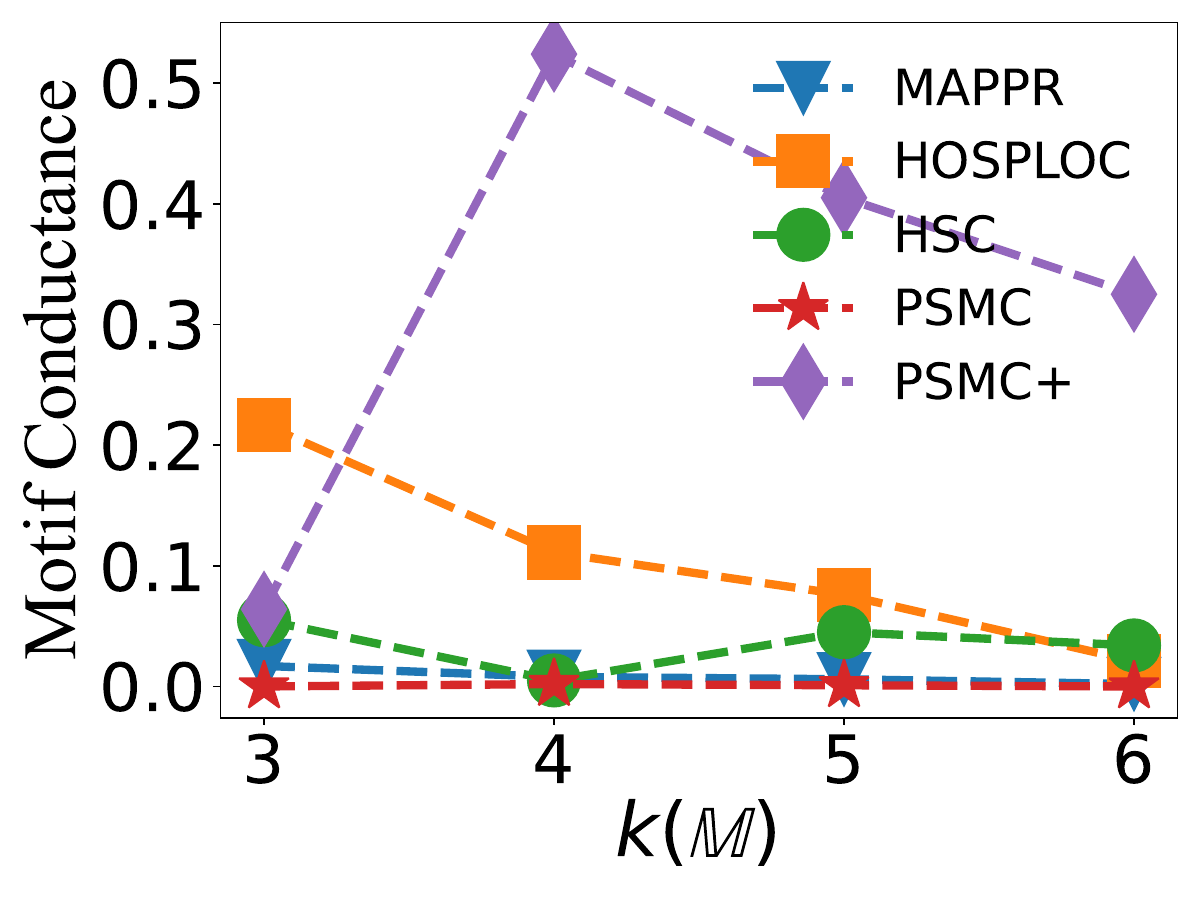}
		\label{fig:mc(b)}
	}\hspace{-0.1cm}
	\subfigure[Youtube]{
		\includegraphics[width=0.19\textwidth]{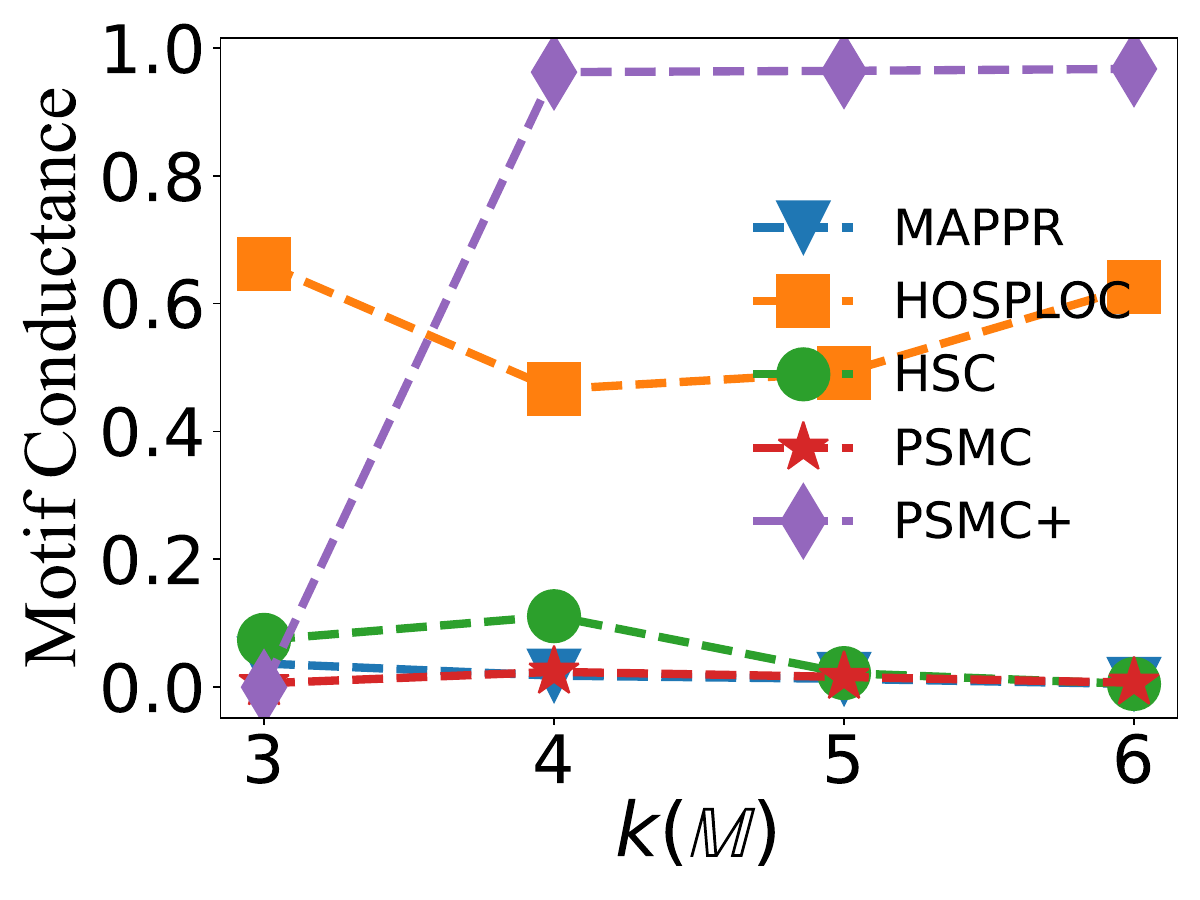}
		\label{fig:mc(c)}
	}\hspace{-0.1cm}
	\subfigure[LiveJ]{
		\includegraphics[width=0.19\textwidth]{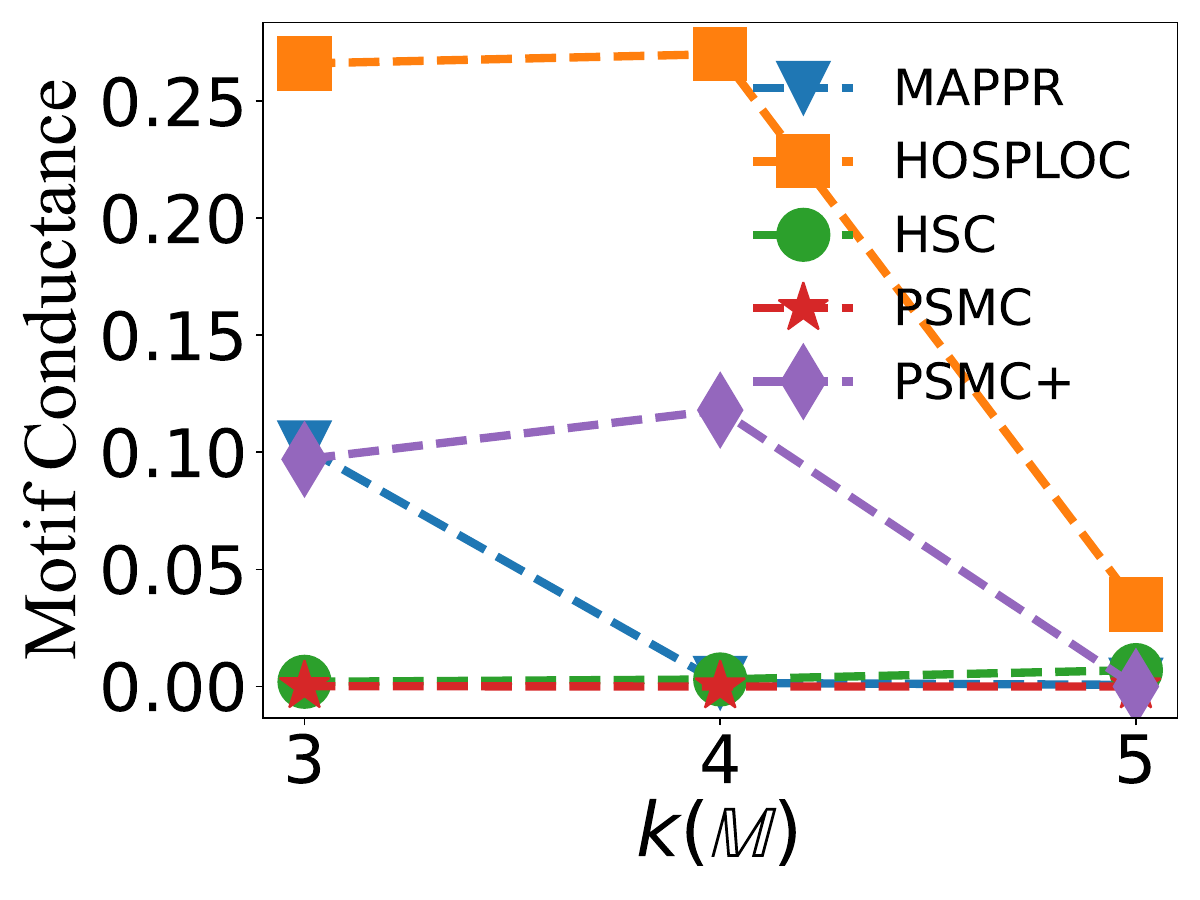}
		\label{fig:mc(d)}
	}\hspace{-0.1cm}
	\subfigure[Orkut]{
		\includegraphics[width=0.19\textwidth]{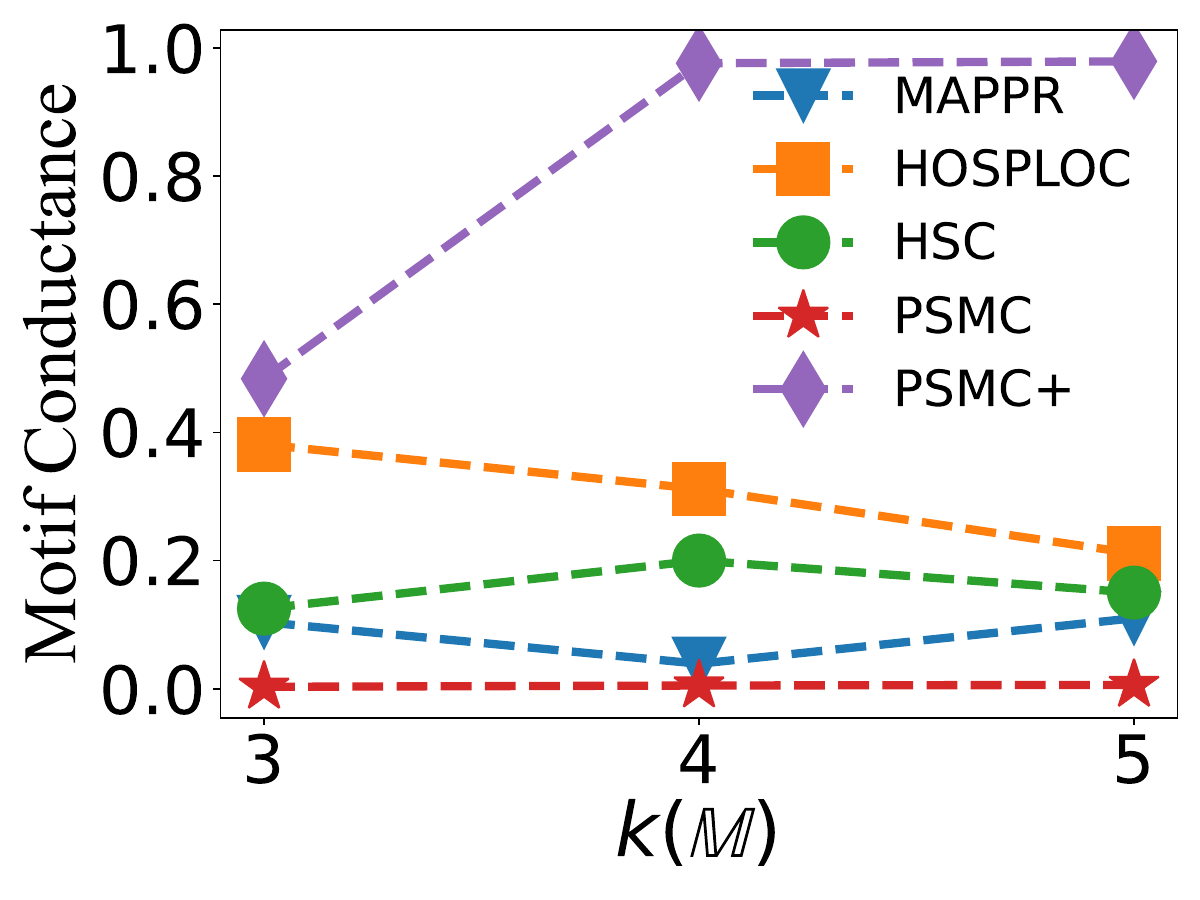}
		\label{fig:mc(e)}
	}
	\vspace{-0.4cm}
	\caption{Quality of  various motif conductance algorithms with varying  $k(\mathbb{M})$ on real-world graphs.}
		\vspace{-0.3cm}
	\label{fig_mc}
\end{figure*}  
	
	\begin{figure}[t!]
		\centering
		\subfigure[\textit{LFR}]{
			\includegraphics[width=0.19\textwidth]{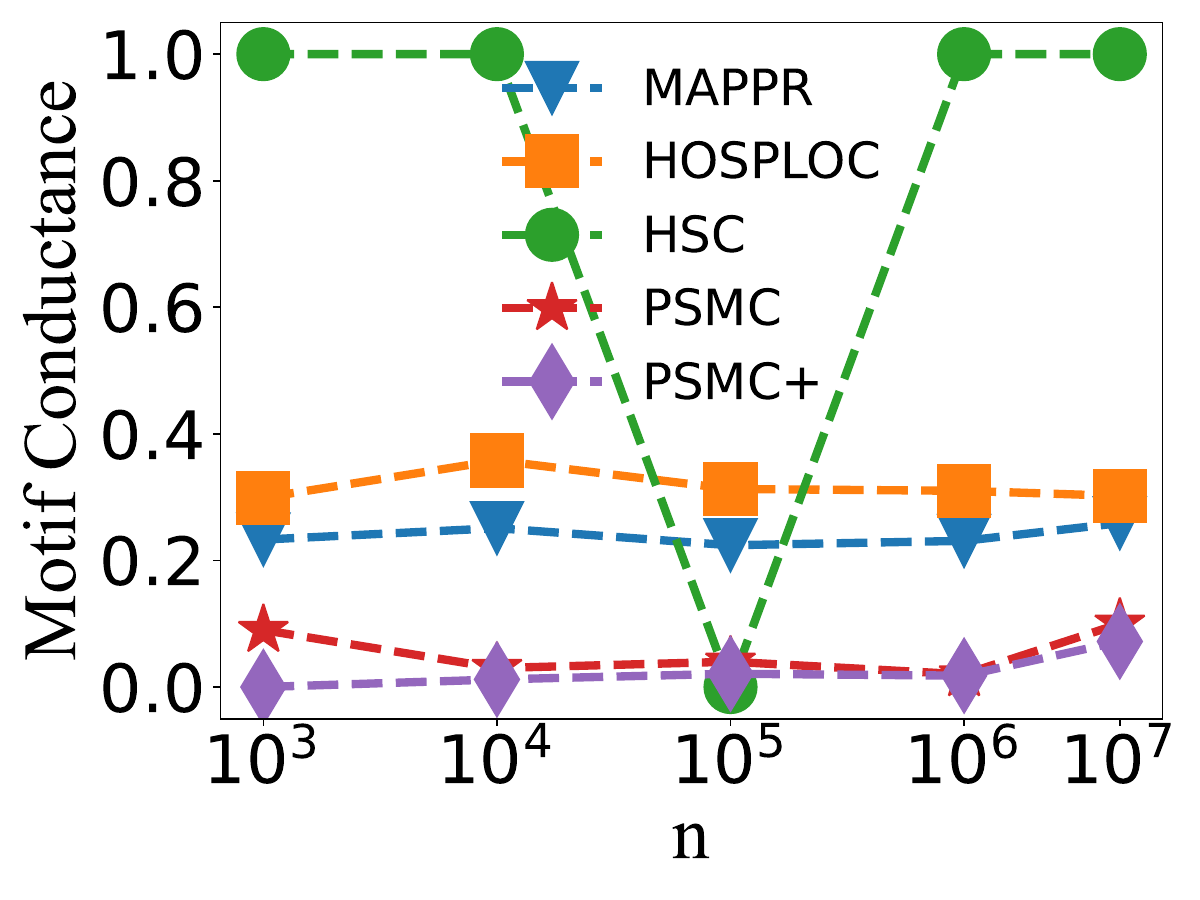}
			\label{fig:mc_syn(a)}
		}
		\subfigure[\textit{PLC}]{
			\includegraphics[width=0.19\textwidth]{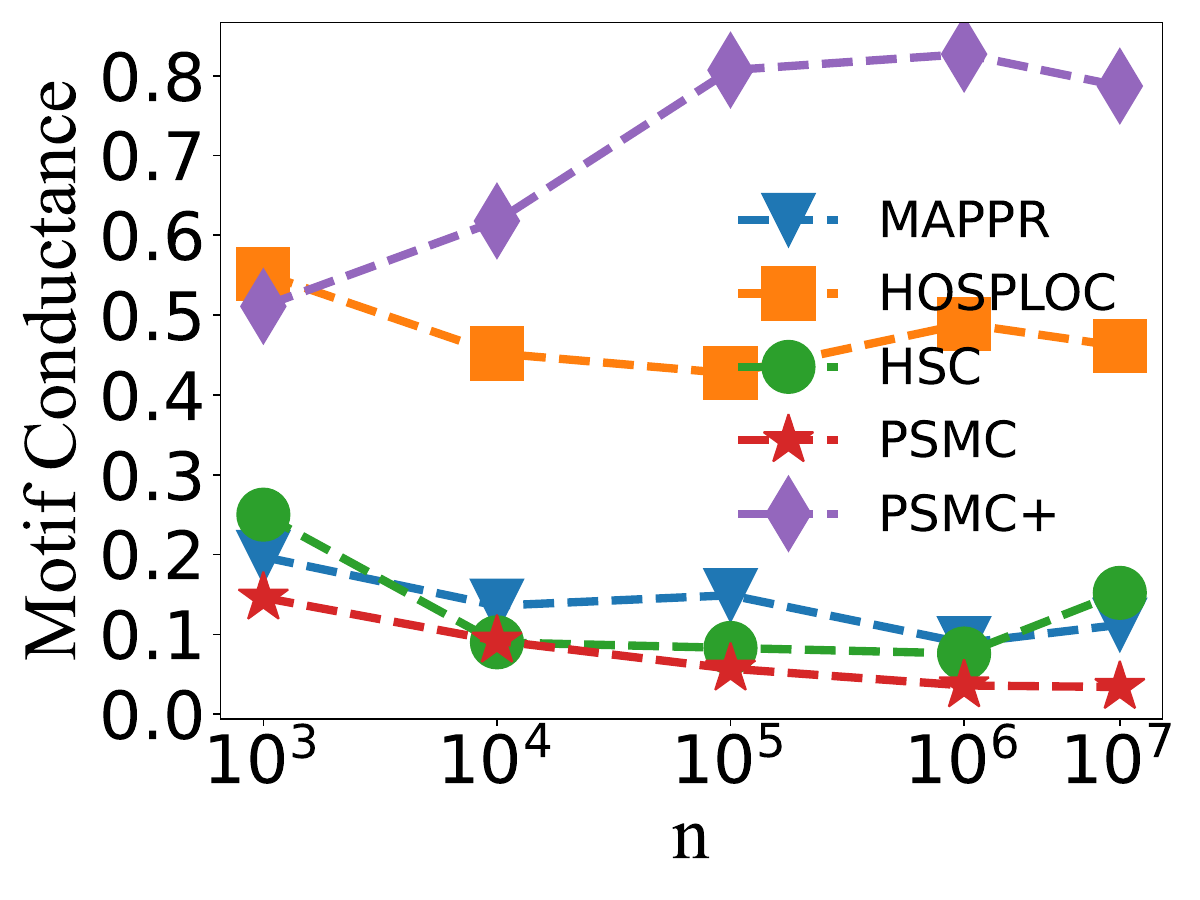}
			\label{fig:mc_syn(b)}
		}
		\subfigure[\textit{ER}]{
			\includegraphics[width=0.19\textwidth]{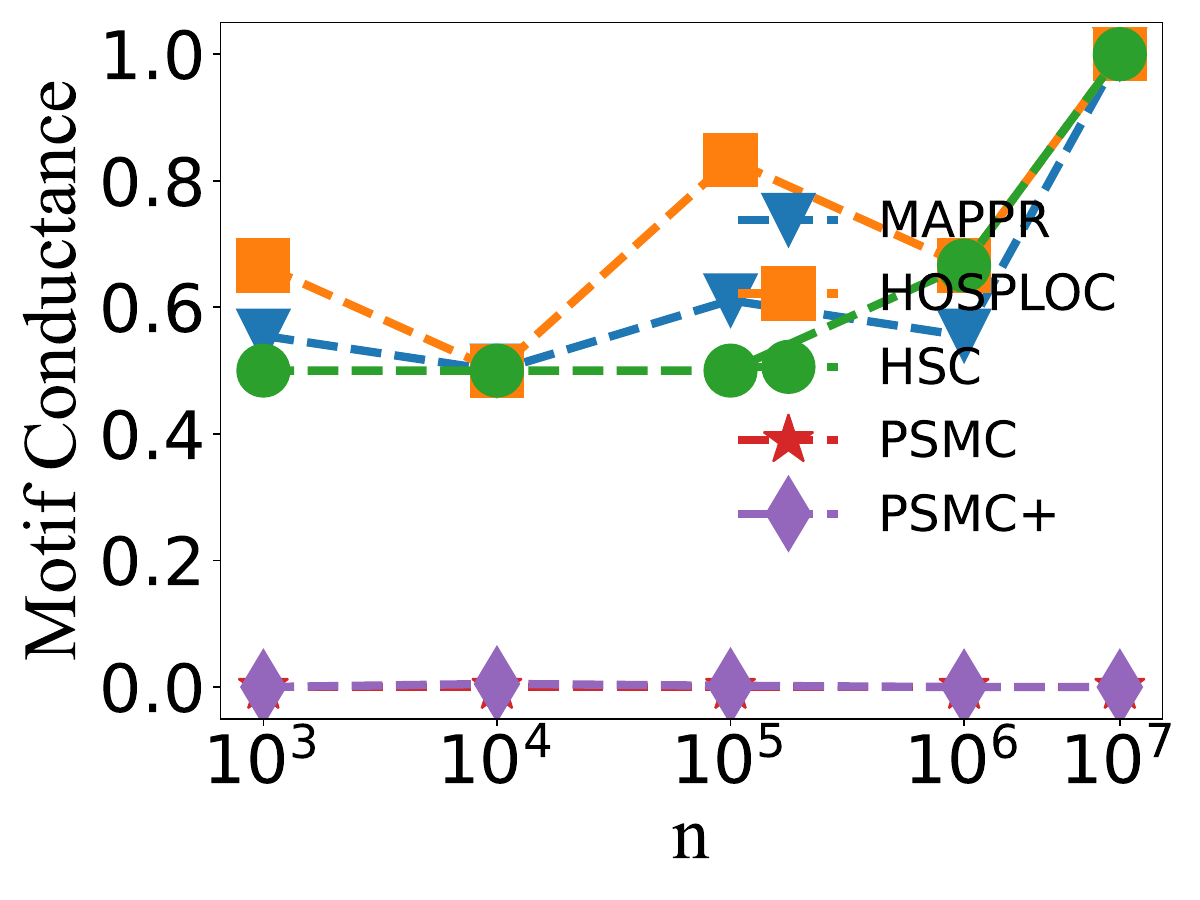}
			\label{fig:mc_syn(c)}
		}
		\subfigure[\textit{BA}]{
			\includegraphics[width=0.19\textwidth]{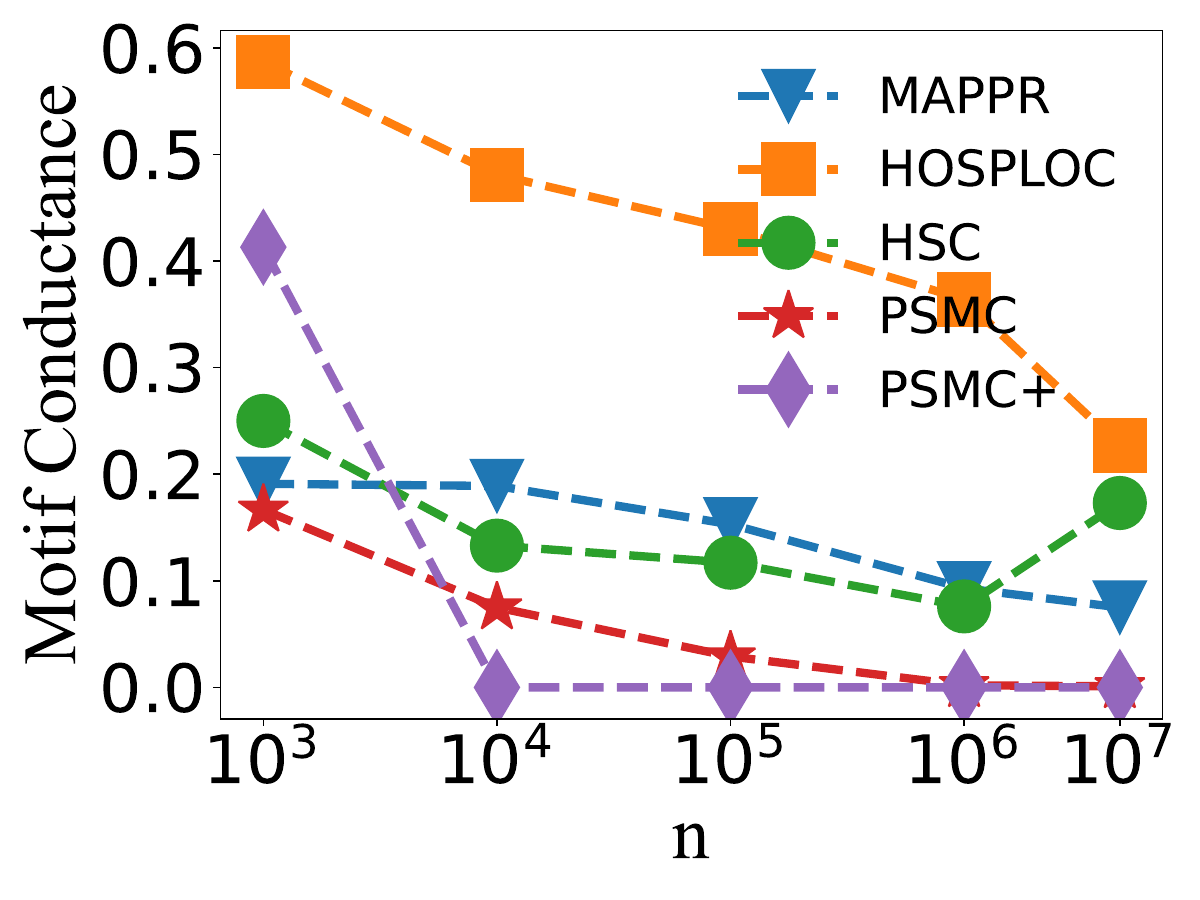}
			\label{fig:mc_syn(d)}
		}
		\vspace{-0.3cm}
		\caption{Quality of  various motif conductance algorithms on synthetic graphs.}	\vspace{-0.3cm}
		\label{fig:mc_syn}
	\end{figure}  
\vspace{-0.3cm}
\subsection{Effectiveness Testing}	
\stitle{Effectiveness Metric.}	We  use  the F1-Score metric  to measure how ``close" each detected cluster $C$ is  to the  ground-truth one. Note that  since F1-Score is  the harmonic mean of precision and recall,  the larger  the  F1-Score, the better the quality of  $C$ \cite{benson2016higher,  DBLP:journals/pvldb/FangYCLL19, DBLP:conf/aaai/HuangCX20}. Besides, we also use  motif conductance (\textit{MC} for short) calculated by Definition \ref{def:mc} to evaluate the quality of the identified community. The smaller the value of the \textit{MC}(C), the better the partition of community $C$. We also report the size of the identified community for completeness. Note that we do not report traditional edge-based metric (e.g., density, conductance) because they are mainly used to measure the quality of communities with edges as atomic clustering units (Section \ref{sec:intro}). 
	
\stitle{Exp-4: Effectiveness  of various graph clustering methods.}  Table \ref{table:metric} only reports these results when the given motif is a triangle, with analogous trends observed across other motifs.  For motif conductance (\textit{MC} for short) metric, we have: (1)  \textit{PSMC} outperforms other methods on four of the five datasets (on Youtube, \textit{PSMC+}  is the champion and \textit{PSMC} is the runner-up).  In particular, \textit{PSMC}  is 167, 12, 20, and 41 times   better than \textit{HSC} on Amazon, Youtube, LiveJ, and Orkut, respectively. This is because \textit{PSMC} can find clusters with near-liner approximation ratio, while \textit{HSC} has quadratic bound (Table \ref{tab:alg}).  (2) \textit{PSMC+}  outperforms \textit{MAPPR} and \textit{HOSPLOC} on four of the five datasets.   This is because  \textit{MAPPR }and \textit{HOSPLOC} are heuristic and have no guarantee of clustering quality  (Table \ref{tab:alg}). However, \textit{PSMC+} is built on top of \textit{PSMC}, so even if \textit{PSMC+} has no theoretical guarantee, it can still get good quality in practice. For F1-Score metric, we have: (1) \textit{PSMC} consistently outperforms other methods (including\textit{ HD} and\textit{ HM}). (2) \textit{SC}, \textit{Louvain}, and \textit{KCore} have poor F1-Scores on most cases. This is  because  they are traditional clustering methods which cannot capture higher-order structural information for graph clustering. For Size metric, on average, the community sizes found by different algorithms from largest to smallest are \textit{SC},  \textit{PSMC+}, \textit{MAPPR}, \textit{HD},  \textit{PSMC}, \textit{HM}, \textit{Louvain}, \textit{HOSPLOC}, \textit{HSC}, and \textit{KCore}. Our algorithm \textit{PSMC} returns the  community size that is ranked in the middle, so it tends to find communities of moderate size. However, other baselines either find the community that is too large or too small, leading to poor interpretability. Thus, these results give clear evidence that our  solutions can indeed find  higher quality clusters when contrasted with baselines.

\stitle{Exp-5: Quality of various motif conductance algorithms.} Figure \ref{fig_mc} depicts the quality of different motif conductance algorithms  with varying $k(\mathbb{M})$ on real-world graphs.  
We have the following observations: (1)  \textit{PSMC} always outperforms other methods under different $k(\mathbb{M})$. Besides,  \textit{PSMC} is almost stable with increasing $k(\mathbb{M})$, while other methods have no obvious change trend with increasing $k(\mathbb{M})$. (2)  \textit{PSMC+} and \textit{HOSPLOC}  fluctuates greatly as $k(\mathbb{M})$ increases. However, \textit{PSMC+}  performs well in   synthetic graphs (see Figure \ref{fig:mc_syn} for details). One possible explanation is that as $k(\mathbb{M})$ increases, the actual effect of estimation bounds in \textit{PSMC+} depends on the type of network (e.g. real-world graphs have poor pruning effects, while synthetic graphs have good pruning effects).  Moreover, we also report these qualities on extensive synthetic graphs. As shown in  Figure \ref{fig:mc_syn},  \textit{PSMC+}  outperforms other methods in most cases. Besides,  \textit{PSMC} is runner up and slightly worse than \textit{PSMC+}. However, the performance of  \textit{MAPPR},  \textit{HOSPLOC}, and \textit{HSC} vary significantly depending on the dataset. For example,  \textit{MAPPR} $<$ \textit{HOSPLOC} $<$  \textit{HSC} on \textit{LFR} synthetic graphs, but   \textit{HSC} $<$ \textit{MAPPR} $<$  \textit{HOSPLOC} on \textit{BA} synthetic graphs, where $A>B$ means $A$ has larger motif conductance. These results indicate that our  algorithms can identify higher quality clusters than the baselines on real-world\&synthetic graphs.

\section{Related Work} \label{sec:related}

\stitle{Traditional Graph Clustering.}  Graph clustering has received much cattention over past decades \cite{10surveycommunity, DBLP:journals/corr/FortunatoH16}.  Modularity \cite{newman_finding_2004, DBLP:conf/sigmod/KimLCY22, DBLP:conf/aaai/ShiokawaFO13} and conductance \cite{von2007tutorial,DBLP:conf/kdd/GleichS12, DBLP:journals/pvldb/GalhotraBBRJ15,DBLP:conf/aaai/LinLJ23,HE2024123915} are two representative models to evaluate the clustering quality of the identified cluster. Informally, they aim to optimize the difference or ratio of edges between the internal and external of the cluster.  However, finding the cluster with optimal  modularity  or  conductance is NP-hard \cite{newman_finding_2004, DBLP:journals/pvldb/GalhotraBBRJ15}. Thus, many heuristic  or approximate  algorithms have been proposed in the literature. For example,  the heuristic algorithm  \emph{Louvain} was proposed to iteratively optimize modularity in a greedy manner \cite{blondel2008fast, DBLP:journals/corr/abs-2311-06047}.  \textit{Fiedler} vector-based spectral clustering algorithm  can output a cluster with a quadratic factor of optimal conductance \cite{DBLP:conf/aaai/LinLJ23}.  Recently, some polynomial solvable cohesive subgraph  models also have been proposed to partition the graph, which are to only optimize the  internal denseness of the identified cluster.  Notable examples include average-degree densest subgraph, $k$-core, and $k$-truss \cite{DBLP:conf/icde/ChangQ19}. However,  these traditional methods mainly focus on the internal or external \emph{lower-order edges} of the cluster, resulting in that cannot capture higher-order structural information for graph clustering. Besides simple graphs, more   complex graphs has also been explored. For example, the graph clustering on attribute graphs \cite{DBLP:conf/www/YangS0HZX21,DBLP:conf/kdd/ZheSX19}, heterogeneous information networks \cite{DBLP:conf/icde/ChenGZJZ19, DBLP:journals/pvldb/JianWC20}, and temporal networks \cite{Chun,DBLP:conf/dasfaa/ZhangLYJ22,DBLP:journals/pvldb/LinYLZQJJ24,DBLP:journals/tbd/LinYLJ22,DBLP:journals/tsmc/LinYLWLJ22}. Obviously, these methods  are orthogonal to our work.

\stitle{Higher-order Graph Clustering.}  In addition to the motif conductance studied in this paper \cite{benson2016higher},  other higher-order graph clustering models also have been proposed in the literature.  For example, motif modularity was proposed to extend the traditional modularity by optimizing the difference between the fraction of motif instances within the cluster and the fraction in a random network  preserving the same degree of vertices \cite{arenas2008motif, DBLP:conf/aaai/HuangCX20}. Higher-order densest subgraph model was proposed where the density  is defined as the number of motif instances divided by the size of vertices \cite{DBLP:journals/pvldb/FangYCLL19, DBLP:journals/pvldb/SunDCS20, DBLP:journals/pacmmod/HeW00023}. 
Li et al.  proposed an edge enhancement approach  to overcome the hypergraph fragmentation issue appearing in the seminal reweighting  framework \cite{DBLP:conf/kdd/LiHWL19}. Unfortunately, they are still essentially optimizing the objective function for traditional lower-order clustering.  Besides simple graphs, higher-order graph clustering on more complicated networks also have been studied, such as heterogeneous information networks \cite{DBLP:conf/kdd/CarranzaRRK20}, labeled networks \cite{DBLP:conf/www/Sariyuce21, DBLP:conf/www/Fu0MCBH23}, multi-layer networks \cite{DBLP:journals/tkde/HuangWC21}, dynamic networks \cite{DBLP:conf/kdd/FuZH20}. Clearly,  these methods  on complicated networks are orthogonal to  our work.

\section{Conclusion}
We first propose a \textit{simple} but \textit{provable}  algorithm  \textit{PSMC} for motif conductance based  graph clustering.  Most notably, \textit{PSMC} can output the result with \emph{fixed} and \emph{motif-independent} approximation ratio, which solves the open question posed by the seminal  two-stage reweighting framework.  We then devise  novel dynamic update technologies and effective  bounds to  further boost efficiency of \textit{PSMC}.  Finally, empirical results  on real-life and synthetic datasets demonstrate the superiority of  the proposed algorithms on both clustering accuracy and running time.

\bibliographystyle{ACM-Reference-Format}
\balance
\bibliography{myreference}

\newpage
\end{document}